\documentclass[12pt,twoside]{article}
\usepackage{amsmath,amssymb,euscript}
\usepackage[numbers]{natbib}
\usepackage{hyperref}

\usepackage[margin=1.3in]{geometry}
\usepackage{amssymb,mathrsfs}
\usepackage{enumerate}

\usepackage{amsmath}
\usepackage{amsthm}
\usepackage{amssymb}
\usepackage{graphicx}

\usepackage{epsf}
\usepackage{epic}
\usepackage{makeidx}
\usepackage{psfrag}
\def\opA{{\mathbb A}}
\def\opH{{\mathbb H}}
\def\opN{{\mathbb N}}
\def\opU{{\mathbb U}}

\def\R{{\mathbb R}}
\def\Z{{\mathbb Z}}
\def\T{{\mathbb T}}

\def\frh{{\mathfrak h}}

\def\cA{{\mathcal A}}
\def\cB{{\mathcal B}}

\def\cD{{\mathcal D}}
\def\cF{{\mathcal F}}

\def\cH{{\mathcal H}}

\def\cL{{\mathcal L}}

\def\cN{{\mathcal N}}

\def\cS{{\mathcal S}}
\def\cU{{\mathcal U}}

\def\cX{{\mathcal X}} 
\def\cY{{\mathcal Y}}

\def\g{{\gamma}}
\def\gam{{\gamma}}
\newcommand{\al}{\alpha}
\newcommand{\s}{\sigma}

\def\om{\omega}

\newcommand{\p}{\partial}

\newcommand{\ra}{\rightarrow}

\newcommand{\ls}{\lesssim}

\newcommand{\im}{\operatorname{Im}}
\newcommand{\re}{\operatorname{Re}}

\def\tr{{\rm Tr}}

\newcommand{\ran}{\rangle}
\newcommand{\lan}{\langle}

\def\1{{\bf 1}}

\def\qf{\om^q}

\def\states{{\mathfrak S}}
\def\qfstates{{\mathfrak Q}}

\def\HHFB{\opH_{hfb}}
\def\Hhfb{\HHFB(\qf_t)}
\def\UHFB{\opU_{\om^q}}

\def\Hd{{H^1}}
\def\Hsdd{{\cH^1_\s}}
\def\cHLd{{\cH^1_\g}}
\def\XL{{X^1}}
\def\SobSymp{{\mathcal H^{\infty,2}}}

\newcommand{\DETAILS}[1]{}

\def\eqnn{\begin{eqnarray*}}
\def\eeqnn{\end{eqnarray*}}
\def\eqn{\begin{eqnarray}}
\def\eeqn{\end{eqnarray}}

\def\prf{\begin{proof}}
\def\endprf{\end{proof}} 

\theoremstyle{plain}
\newtheorem{theorem}{Theorem}[section]
\newtheorem{definition}[theorem]{Definition}
\newtheorem{proposition}[theorem]{Proposition}

\newtheorem{lemma}[theorem]{Lemma}
\newtheorem{corollary}[theorem]{Corollary}
\newtheorem{remark}[theorem]{Remark}

\begin{document}

\parskip=8pt


\title{The time-dependent Hartree-Fock-Bogoliubov equations for Bosons. Revision
  \footnote{This is a revision of an earlier manuscript, \cite{BachBreteauxChenFroehlichSigal2016} 
 aiming at correcting an error which crept into one of the estimates entering the proof of the local existence result.}} 

\newcommand{\DATUM}{11-May-2018}         
\pagestyle{myheadings}                         
\markboth{\hfill{The time-dependent Hartree-Fock-Bogoliubov Equations, \DATUM}}{{The time-dependent Hartree-Fock-Bogoliubov equations, 
 \DATUM}\hfill}  %


\author{V. Bach, S.  Breteaux, Th. Chen, J. Fr\"ohlich, I. M. Sigal}

\DETAILS{\author[V. Bach]{Volker Bach}
\address[V. Bach]{
Institut f\"ur Analysis und Algebra
Carl-Friedrich-Gauss-Fakult\"at,
Technische Universit\"at Braunschweig,
38106 Braunschweig,
Germany}
\email{v.bach@tu-bs.de}

\author[S. Breteaux]{S\'ebastien 
Breteaux}
\address[S. Breteaux]{BCAM - Basque Center for Applied Mathematics, 48009 Bilbao, Basque-Country, Spain}
\email{sbreteaux@bcamath.org}

\author[T. Chen]{Thomas Chen}
\address[T. Chen]{Department of Mathematics, University of Texas at Austin, Austin TX 78712, USA}
\email{tc@math.utexas.edu}

\author[J. Fr\"ohlich]{J\"urg Fr\"ohlich}
\address[J. Fr{\"o}hlich]{Institut f{\"u}r Theoretische Physik, ETH H{\"o}nggerberg, CH-8093 Z{\"u}rich, Switzerland}
\email{juerg@phys.ethz.ch}

\author[I.M. Sigal]{Israel Michael Sigal}
\address[I.M. Sigal]{Department of Mathematics, University of Toronto, Toronto, ON M5S 2E4, Canada}
\email{im.sigal@utoronto.ca}}
\maketitle

\begin{abstract}
We introduce the map of dynamics of quantum Bose gases into dynamics of  quasifree states, which we call the ``nonlinear quasifree approximation". We use this map to derive the time-dependent Hartree-Fock-Bogoliubov (HFB) equations describing the dynamics of quantum fluctuations around a Bose-Einstein condensate.
We prove global well-posedness of the HFB equations for pair potentials satisfying suitable regularity conditions, and we establish important conservation laws. 
We show that the space of solutions of the HFB equations has a symplectic structure reminiscent of a Hamiltonian system. 
This is then used to relate the HFB equations to the HFB eigenvalue equations discussed in the physics literature. We also construct Gibbs equilibrium states at positive temperature associated with the HFB equations, and we establish criteria for the appearance of Bose-Einstein condensation.
\end{abstract}

\section{Introduction}
\label{sec-intro-1}

In this paper,  we derive the time-dependent \textit{ Hartree-Fock-Bogoliubov (HFB) equations} describing  quantum fluctuations of a non-relativistic Bose gas around a Bose-Einstein condensate and study their properties. 

\subsection{Quantum Many-Body Problem} \label{sec:quant-mb-probl}

The starting point of our analysis is a second-quantized description of a quantum-mechanical many-body system of Bose point particles (bosonic atoms). We first consider systems of finitely many particles. The Hilbert space of pure state vectors is given by the bosonic Fock space 
\eqn \label{Fock}
    \cF := \bigoplus_{n=0}^\infty \frh^{\otimes_{sym}n} \,,
\eeqn 
where the $n$-th summand is an $n$-fold symmetric tensor product of the one-particle Hilbert space
$$\frh := L^2 (\R^{d}),$$
accounting for the Bose-Einstein statistics of the particles.
The time evolution of the system is generated by the quantum Hamiltonian  
\eqn  \label{model-ham} 
    \opH = \int dx \; \psi^*(x)h\psi(x) 
    + \frac{1}{2}\int dx dy \; v(x-y)\psi^*(x) \psi^*(y)
    \psi(x) \psi(y) \,,
\eeqn
where, in the position-space representation, the operator $h$ is given by 
$$h:= -\Delta +V(x),\quad x\in \mathbb{R}^{d},\, d=1,2,3, \dots,$$ 
with $\Delta$ the Laplacian acting on $\frh$. In \eqref{model-ham}, $\psi(x)$ and $\psi^*(x)$ denote annihilation-
and creation operators, respectively. These operators (actually operator-valued distributions)
satisfy the canonical commutation relations (CCR): 
\eqn\label{commrel}
    \big[\psi(x),\psi^*(y)\big]=\delta(x-y) \,,\qquad 
    \big[\psi(x),\psi(y)\big]=0\; = \; 
    \big[\psi^*(x),\psi^*(y)\big]\,,
\eeqn 
see, e.g., \cite{BratteliRobinson-II-1996}. We will write $\psi^\sharp(x)$ for either $\psi(x)$ or $\psi^*(x)$.

We always impose the following conditions:
\begin{align} \label{cond-a}
(a) \quad    &  \text{The external potential $V$ is infinitesimally bounded}
\nonumber \\ &  \text{with respect to the Laplacian $-\Delta$.}
\\[2ex] \label{cond-b}
(b) \quad    &  \text{The pair potential $v$ is even, $v(x)=v(-x)$,}
\nonumber \\ &  \text{and relatively bounded with respect to $\Delta$.}
\end{align}
These conditions imply that $\opH$ is self-adjoint on the domain of
the operator 
\begin{equation} \label{free-ham}
\opH_0 := \int dx \; \psi^*(x)(-\Delta)\psi(x)
\end{equation} 
(see Appendix~\ref{sec:Self-adjointness_H}). We note that 
these conditions allow both $V$ and $v$ to have Coulomb singularities.

Let $W^{p, r}(\R^d)$ denote the standard Sobolev space over $\R^d$.  In Section \ref{sec-GWP-HFB-1} we will use a stronger condition on $v$: 
\begin{align} \label{cond-b'}
(b') \quad   &  \text{The pair potential $v$ is even, $v(x)=v(-x)$,} 
\hspace*{15mm}
\nonumber \\ &  \text{and satisfies $v \in W^{p, 1}$, for some $p>d$.}
\end{align}

States of the system are normalized positive linear (`expectation') functionals, $\om$, on the Weyl algebra ${\frak W}$ over Schwartz space $\cS(\R^d)$, which is generated by Weyl operators, $$W (f):=e^{i\phi (f)}, \text{   with   }\, \phi (f):=\psi^* (f)+\psi(f), \text{    and   }\, f \in \cS(\R^d),$$ 
(see \cite{BratteliRobinson-II-1996}, Section 5.2.3).   The set of states is denoted by $ \states$.

States correspond either to a finite number of Bose particles, as in the case of BEC experiments in traps, or to an infinitely extended gas at a non-zero particle density and some fixed temperature. States of finitely many particles are given by density operators on Fock space $\cF$: $\om (\opA)=\tr (\opA D)$, for all bounded operators 
$\opA$ on $\mathcal{F}$, (and in particular for all elements ${\mathbb{A} \in \frak W}$), where $D$ is a positive, trace-class operator on $\cF$ with unit trace. 


 It will be convenient to consider states defined on arbitrary products,
$$\psi^{\#}(f_1) \ldots \psi^{\#}(f_n),$$ of creation- and annihilation
operators. Expectations of such products in a state $\omega$, henceforth called \textit{correlation functions}, can be defined by applying partial derivatives, 
$\partial_{s_k}$, to expectation values 
$$\om(W(s_1f_1) \ldots W(s_nf_n))$$ of Weyl operators.  
We will only consider states with the property that these derivatives, and hence the corresponding correlation functions, exist, for arbitrary $n$; such states are henceforth called \textit{``regular states''}. Of particular interest to us are correlation functions with $n\le 4$. Their existence is guaranteed by assuming, e.g., that $\om (\opN^2)<\infty$, where $\opN$ is the particle number operator, $\opN:= \int dx \; \psi^*(x) \psi(x)$. This assumption implies, in particular, that $\om$ is given by a density operator on $\mathcal{F}$. (We remark, however, that existence of correlation functions follows from considerably weaker assumptions; e.g., from an appropriate version of the assumption that the particle density in the gas is finite.)
 
 The multilinear functionals \mbox{$\om(\psi^{\sharp}(f_1) \cdots \psi^{\sharp}(f_n)), \, f_i \in \cS(\R^d), \forall i,$} are given by tempered distributions (this is the nuclear theorem), which we formally write as $$\om(\psi^{\sharp}(x_1) \ldots \psi^{\sharp}(x_n)).$$ By an ``observable'', we refer either to an element of the Weyl algebra ${\frak W}$ or to a linear combination of operators of the form $\psi^{\#}(f_1) \ldots \psi^{\#}(f_n)$. (We remark that the term ``observable'' is however usually reserved for products 
 $\psi^{\#}(f_1) \ldots \psi^{\#}(f_n)$ that are gauge-invariant, i.e., invariant under phase transformations, $\psi \mapsto e^{i\theta} \psi, \psi^{*} \mapsto e^{-i \theta} \psi^{*}$.)
 
 The time-evolution of regular states is  given by the 
von~Neumann-Landau equation \cite{vonNeumann1927, Landau1927} (see also \cite{Pauli1928, Bloch1946}, and \cite{TerHaar1969}) for some history)
\eqn\label{eq-vNeum-1} 
    i\partial_t\om_t(\opA) = \om_t([\opA,\opH]) \,,    
\eeqn
for arbitrary observables $\opA$, which extends the standard von~Neumann-Landau equation to general $C^*-$algebras (see e.g. \cite{BratteliRobinson-II-1996, MerkliMueckSigal2007}). 

\subsection{Quasifree States and Truncated Expectations}

Since the evolution equation \eqref{eq-vNeum-1} is extremely complicated to analyze, one is interested in manageable approximations to it. Our approximation consists of restricting the dynamics given by \eqref{eq-vNeum-1} to \textit{quasifree states},  
  the simplest - yet sufficiently rich - class of states generalizing the Hartree and Hartree-Fock ones, on one hand and the Gaussian random processes, on the other, as has been first realized and used in \cite{BachLiebSolovej1994}.\footnote{The notion of quasifree states was introduced in \cite{Robinson1965}; see 
\cite{BratteliRobinson-II-1996} and references therein.}

Quasifree states are defined in terms of truncated expectations, which we define next. We use the short-hand notation
$\psi_j  :=\psi^{\sharp_j}(x_j)$.  The $n^{th}$ order \textit{truncated expectations (correlation functions)},
$\om^T(\psi_1, \ldots, \psi_n),$ 
of a state $\omega$ are defined recursively through 
\begin{equation}
\label{eqn:quasi-wick} 
\omega(\psi_1 \cdots \psi_n)
 =
\sum_{P_n} \prod_{J \in P_n}
 \om^T(\psi_{i_1}, \ldots, \psi_{i_{\vert J\vert}})\,,
\end{equation}
where the $P_n$ are partitions of the ordered set $\{1, ..., n\}$ into
ordered subsets, $J$. 
The simplest examples of truncated (or connected) correlations are
\eqn\label{2term-wick}
     \om^T(\psi(x)) &=&\omega(\psi(x)) \,,  \nonumber\\
     \om^T(\psi_1, \, \psi_2) &=&   \omega(\psi_1 \psi_2) -
    \omega(\psi_1) \, \omega(\psi_2) \,.
\eeqn
%

A state $\omega$ is called \textit{quasifree} if truncated n-point expectations vanish for $n>2$, i.e.,
\begin{equation}\label{quasifree}
 \om^T(\psi_1, \ldots, \psi_n) = 0, \quad \forall n>2,
 \end{equation}
We denote quasifree states by $\om^q$ and 
the set of quasifree states by $\qfstates \subset \states$.

It follows from the definition that all $n$-point expectations, $\qf(\psi_1^{\sharp_1}\cdots \psi_n^{\sharp_n})$, 
with $n>2$, in a quasifree state $\omega^{q}$
 can be expressed in terms of $\qf(\psi^{\sharp_i}_i)$ and
$\qf(\psi^{\sharp_j}_j\psi^{\sharp_k}_k)$, with $i,j,k\in\{1,\dots,n\}$. The explicit formula is called Wick's formula, or \textit{Wick's theorem}; see \cite{BratteliRobinson-II-1996}. 
Examples for small orders are given in Appendix~\ref{app-quasifreedef}.

Given an arbitrary, not necessarily quasifree state 
$\om \in \states$, with $\om(\opN) < \infty$, there exists a unique 
quasifree state, denoted $q[\om] \in \qfstates$, such that expectations
\eqn\label{2term-wick.01}
\om(\psi_1^{\sharp_1}) \ = \ q[\om](\psi_1^{\sharp_1})
\quad \text{and} \quad
\om(\psi_1^{\sharp_1} \, \psi_2^{\sharp_2})
\ = \ 
q[\om](\psi_1^{\sharp_1} \, \psi_2^{\sharp_2})
\eeqn
of quadratic or lower order agree (see Subsection~\ref{sec:truct-expect} below). We call the state $q[\om]$ the
quasifree reduction of $\om$.\footnote{This notion was introduced in \cite{ArakiShiraishi1971} (see below). For a related notion in the context the gauge invariant twice differentiable states, see \cite{OhyaPetz1993}. For the definition of the gauge invariant states, see Subsection \ref{sec:qf-dyn} below.} The map $q: \states \to
\qfstates$ is idempotent, $q \circ q = q$, and acts as a projection of
the convex space $\states$ of all states onto the space of quasifree states
$\qfstates$.

\subsection{Quasifree Dynamics}\label{sec:qf-dyn}

As mentioned above, detailed properties of the dynamics of a
many-body system described by the von Neumann-Landau equation \eqref{eq-vNeum-1} are
difficult to unravel, and approximations are therefore needed to extract interesting
qualitative features. 

The main idea is to restrict the dynamics to quasifree states.
However, the property of being quasifree is not preserved by the dynamics given by \eqref{eq-vNeum-1} and the main question here is how to map the true quantum evolution onto the class of quasifree states.

The {\it effective} dynamics we propose replaces equation \eqref{eq-vNeum-1},  with 
an initial condition  $\om_0 \in \states$,  by the equation 
\eqn\label{eq-vNeum-quasifree}
    i\partial_t\qf_t(\opA) = \qf_t\big( [\opA,\opH] \big) \, ,  \quad \text{ with } \omega^{q}_{t=0}=q[\omega_{0}]\,,
\eeqn
for all observables $\opA$ that are at most \textit{quadratic} in
creation- and annihilation operators. 

For the Hamiltonian $\opH$ given
by \eqref{model-ham}, the commutator $[\opA,\opH]$ contains products
of at most four creation- and annihilation operators; their expectation in $\omega^{q}_{t}$ is then evaluated by using Wick's theorem for the quasi-free state $\omega^{q}_{t}$. 

 In contrast to the von Neumann-Landau equation \eqref{eq-vNeum-1}, the {\it quasifree dynamics} \eqref{eq-vNeum-quasifree} is \textit{non-linear}.

Of course, one expects the effective evolution to be close the original one only if $\omega_{0}$ is close to $q[\omega_{0}]$ in an appropriate sense. 
We emphasize that, in general, $$\qf_t \neq q[\om_t],$$ 
even if the initial state $\om_0 = q[\om_0] \in \qfstates$ is quasifree. 
That is, the trajectory of quasifree states $\omega^{q}_{t}$ determined by 
\eqref{eq-vNeum-quasifree} is \textit{not} the projection, $q$, of
the trajectory $\omega_{t}$ of states evolving according to the full dynamics in \eqref{eq-vNeum-1} 
onto the space $\qfstates$ of quasifree states.

 We call \eqref{eq-vNeum-quasifree} the  {\it nonlinear quasifree approximation} (it was called the quasifree reduction in \cite{BachBreteauxChenFroehlichSigal2016}.)

The deviation of a state $\om \in \states$ from its quasifree reduction
 $q[\om] \in \qfstates$ can be quantified in terms of their relative 
entropy $S_{\mathrm{rel}}(\om , q[\om]) := 
\tr\big\{ D_\om \big( \ln[D_\om] -  \ln[D_{q[\om]}] \big) \big\}$,
provided $\om$ and hence $q[\om]$ are given by density operators $D_\om$ and
$D_{q[\om]}$, respectively \cite{GottliebMauser2007}. In fact, 
$S_{\mathrm{rel}}(\om, q[\om])$ may be viewed as the distance of
$\om$ to $\qfstates$, since $S_{\mathrm{rel}}(\om, \om') \geq 0$
with equality if, and only if, $\om = \om'$ and
\begin{align} \label{RelativeEntropy}
S_{\mathrm{rel}}(\om, q[\om])
\ = \
\inf_{q \in \qfstates} S_{\mathrm{rel}}(\om, q) .
\end{align}

 It has been shown in \cite{BenedikterSokSolovej2018} that for pure states the quasifree dynamics as defined above (\cite{BachBreteauxChenFroehlichSigal2016}) is a consequence of the Dirac-Frenkel principle,  in which the right side of the von Neumann-Landau equation \eqref{eq-vNeum-1} is projected onto a selected class of states. 


We will show that equation \eqref{eq-vNeum-quasifree} is 
 equivalent to the nonlinear, self-consistent 
evolution equation
\eqn\label{eq-qf-consistency-0}
i\partial_t \qf_t(\opA) 
\ = \ 
\qf_t\big( [\opA,\opH_{\rm hfb}(\qf_t)] \big) ,
\eeqn
for all observables $\opA$, 
where $\opH_{\rm hfb}(\qf)$ is an explicit \textit{quadratic} 
Hamiltonian given in Eq. \eqref{eq-HHFB-def-1}, which depends on a quasifree state $\omega^{q}$; see Theorem~\ref{thm-Hquadr-dyn-1}.  The equivalence holds for observables linear or quadratic in creation- and annihilation operators. 

 Equation \eqref{eq-vNeum-quasifree}, with the Hamiltonian $\opH$ given
by \eqref{model-ham}, is equivalent to  the HFB equations \eqref{eq:BHF-phi}-\eqref{k} derived from it below.

For $U(1)$-gauge invariant quasifree states, i.e., states, $\qf$, satisfying $\qf(\psi)=\qf(e^{i\theta}\psi), \forall \theta$, etc., Eq.  \eqref{eq-vNeum-quasifree}, with $\opH$ given
by \eqref{model-ham},  and consequently self-consistent equation \eqref{eq-qf-consistency-0} and the HFB equations \eqref{eq:BHF-phi}-\eqref{k}, reduce to the bosonic Hartree-Fock equation.
 Indeed, gauge-invariant quasifree states have vanishing truncated expectations $\phi_{\qf}$ and $\sigma_{\qf}$, as follows from the  $U(1)$-gauge invariance, 
 which implies that $\qf(\psi)=\qf(e^{i\theta}\psi)=e^{i\theta}\qf(\psi)$, and hence $\phi_{\qf}=0$, and similarly one shows that $\sigma_{\qf}=0$.

\subsection{HFB Equations for Truncated Expectations}\label{sec:truct-expect}

As was mentioned above, a quasifree state $\om^q \in \qfstates$ determines, and 
is determined by, the truncated expectations up to second order in the 
following sense
\begin{itemize}
\item[1.] \underline{$\om \to \Gamma$:} Given a (not necessarily quasifree)
state $\om \in \states$ and its expectations
\begin{equation} \label{omq-mom} 
\begin{cases}
    \phi(x) & :=\om [\psi(x)],\\
    \gamma (x;y) & :=   \om[\psi^{*}(y) \, \psi(x)] - \om [\psi^{*}(y)] \, \om [\psi(x)] ,\\
    \sigma (x,y) & :=   \om[\psi(x) \, \psi(y)] - \om [\psi(x)] \, \om [\psi(y)] \, ,
\end{cases} 
\end{equation}
up to second order, and denoting by $\gamma$ and $ \sigma$ the operators
with integral kernels given by $\gamma(x,y)$ and $\sigma(x,y)$,
respectively, we have that (see \eqref{Gam-posit} below)
\begin{align} \label{Gam-pos}
\Gamma = \begin{pmatrix}
\gamma & \sigma\\
 \bar{\sigma} & 1+\bar \gamma
\end{pmatrix} \geq 0 \,, 
\end{align}
where $\bar A := C \sigma C$, with $C$ denoting complex conjugation in
the position-space representation, (i.e., complex conjugation of wave
functions of spatial variables). Note in passing that this implies, in
particular, that
\begin{equation} \label{gam-sig-cond}
\gamma=\gamma^*\ge 0\ \text{ and }\ \sigma^*=\bar\sigma .
\end{equation} 

\item[2.] \underline{$\Gamma \to \om^q$:} Conversely, given 
$\gamma = \gamma^* \geq 0$ and $\sigma^* = \bar\sigma$ such that
$\Gamma := \big( \begin{smallmatrix}
\gamma & \sigma\\
 \bar{\sigma} & 1+\bar \gamma
\end{smallmatrix}\big) \geq 0$ obeys \eqref{Gam-pos} and 
$\phi \in L^2(\mathbb{R}^d)$, there exists a unique quasifree state 
$\om^q \in \qfstates$ such that \eqref{omq-mom} holds true with
$\om^q$ replacing $\om$. 

Actually, the condition that $\phi \in L^2(\mathbb{R}^d)$ is too restrictive
and can be relaxed, depending on the context.

\item[3.] \underline{$\om \to q[\om]$:} Given a state $\om \in \states$
and going through 1.\ and 2.\ above yields the quasifree reduction
$q[\om] := \om^q$ of $\om$.
\end{itemize}

%
The matrix operator in \eqref{Gam-pos} is called ``generalized
one-particle density matrix''. The positivity condition on $\Gamma$ in
\eqref{Gam-pos} can be expressed directly in terms of $\g$ and $\s$;
see \cite{BachBreteauxKnoerrMenge2014}, \cite{BachBreteauxChenFroehlichSigal2016}. 
The steps \textit{1.}\
  and \textit{2.}, whose composition yields the quasifree reduction
  $q$, were first carried out in \cite[Lemmata~3.2-3.5]{ArakiShiraishi1971}.

 We will use \eqref{Gam-pos} in proving the global existence for the HFB equations (see Proposition~\ref{prp-gampos-sigsym-1}(4) and the pragraph after Eq. \eqref{gam-E-bnd}). 
\DETAILS{ 
\eqref{omq-mom} implies that 
 \begin{equation} \label{Gam-pos}\Gamma=\begin{pmatrix}
\gamma & \sigma\\
 \bar{\sigma} & 1+\bar \gamma
\end{pmatrix}\geq 0\,.
\end{equation}
The converse is shown in  \cite{BachLiebSolovej1994} for fermion states and in \cite{Solovej2014} for pure boson ones, with $\tr\g<\infty$ in both cases. 
In the following, we always assume that \eqref{Gam-pos} holds.}

Evaluating \eqref{eq-vNeum-quasifree} 
for monomials $\opA\in\mathcal A^{(2)}$, where 
\[
\mathcal A^{(2)}
\ := \ 
\big\{ \psi(x), \: \psi^*(x)\psi(y), \; \psi(x)\psi(y) \big\},
\] 
yields a system of coupled nonlinear PDE's for
$(\phi_t,\gamma_t,\sigma_t)$, the {\it Hartree-Fock-Bogo\-liubov
  (HFB) equations}, presented in Eqs.~\eqref{eq:BHF-phi}-\eqref{k},
below. Since quasifree states are characterized by their truncated
expectations $\phi$, $\gamma$ and $\sigma$, this system of equations is
equivalent to Eq.~\eqref{eq-vNeum-quasifree}.

To give a first impression of the HFB equations \eqref{eq:BHF-phi}-\eqref{k}, 
we formally assume the pair interaction potential to be a
delta distribution, $v(x)=g\delta(x)$, where~$g\geq0$ is a coupling
constant. The HFB equations then have the form
\begin{align}
    i\partial_{t}\phi_t & =h_{g\delta}(\gamma_t^{\phi_t})\phi_t +g d(\sigma_t^{\phi_t}) \bar\phi_t-2g|\phi_t|^{2}\phi_t \,,
    \label{eq-HFBdelt-1-1}\\
    i\partial_{t}\gamma_t & =[h_{g\delta}(\gamma_t^{\phi_t}),\gamma_t]+g d(\sigma_t^{\phi_t}) \sigma_t^*-g \sigma_t \overline{d(\sigma_t^{\phi_t})}\,,
    \label{eq-HFBdelt-1-2}\\
    i\partial_{t}\sigma_t & =[h_{g\delta}(\gamma_t^{\phi_t}),\sigma_t]_{+}+ g [d(\sigma_t^{\phi_t}), \gamma_t]_+
    + d(\sigma_t^{\phi_t})\,,
    \label{eq-HFBdelt-1-3}
\end{align}
where $[A,B]_+=AB^T+BA^T$, with  :$A^T:=C A^* C$, 
 $d(\sigma)(x) :=\sigma(x, x)$, $d(\gamma)(x):=\gamma(x,x)$, and
\begin{align}
\label{eq-nmwh-notations-0} 
   & \sigma^\phi:=\sigma+\phi\otimes \phi \,, \quad  \gamma^\phi:=\gamma+|\phi\ran\lan\phi| \,,  \\ 
   & h_{g\delta}(\gam):=h+2g d(\gam) \,,
\end{align}
with $h$ as in Eq.\,\eqref{model-ham}. Here and in what follows, we denote the {\it multiplication operators} and the {\it functions} by which they multiply by the {\it same symbols}. The meaning will always be clear from context.

The physical interpretation of the truncated expectations of $\qf_t$
is as follows: The function $\phi_t$ is the quantum-mechanical one-particle wave function of the
Bose-Einstein condensate, while $\gamma_t$ and $\sigma_t$ describe the
dynamics of sound waves in the quasifree approximation; in particular,
$d(\gamma_t)$ determines the density of the ``thermal cloud'' of atoms. (In the
physics literature, $n=d(\gamma)$ and $m=d(\sigma)$ are called the
\textit{non-condensate density} and \textit{anomalous density},
respectively.)

The HFB equations \eqref{eq-HFBdelt-1-1}, \eqref{eq-HFBdelt-1-2} and \eqref{eq-HFBdelt-1-3} provide a {\it time-dependent extension} of the
standard \textit{stationary} Hartree-Fock-Bogoliubov equations for a Bose gas
found in the physics literature; see,
e.g., \cite{BlaizotRipka1986, DoddEdwardsClarkBurnett1998, Griffin1996, ParkinsWalls1981}.  Related
equations (with $\phi_t=0$) appear in superconductivity. 
These so-called \textit{Bogoliubov-de~Gennes} equations are equivalent to the
\textit{BCS} effective Hamiltonian description.

\subsection{Summary of Main Results} 
The formulation of the nonlinear quasi-free approximation in the form of equation\,
\eqref{eq-vNeum-quasifree}, and the derivation of its equivalent formulations as self-consistent equation \eqref{eq-qf-consistency-0} for $\qf_t$ and 
 the HFB equations \eqref{eq:BHF-phi}-\eqref{k}) for the truncated expectations $\phi$, $\gamma$, $\sigma$, 
are among the main results presented in this paper; (see Theorems~\ref{prop-HFB-vfull-1} and \ref{thm-Hquadr-dyn-1}).

We also initiate a mathematical study of solutions of the HFB equations. 
 In particular, if the initial state $\omega_0$ is s.t. the operator $\gam_0$ is 
 trace-class (i.e., the number of atoms is finite)  and $\sigma_0$ is Hilbert-Schmidt -- for precise hypotheses 
see Section~\ref{sec-HFB-general-1} -- we have the following results:
\vspace*{-1ex}
\begin{itemize}
\item Conservation of the total number of atoms in the gas:
\eqn
\cN(\phi_t,\gamma_t,\sigma_t) \ := \ \qf_t(\opN) \, , 
\eeqn 
where $\opN$ is the particle-number operator; (see
Corollary~\ref{cor-partnumb-en-conserv}).

\item Existence and conservation of the total energy (under suitable conditions 
on the two-body potential $v$ and on the initial state $\qf_0$):
\eqn 
    \mathcal{E}(\phi_t,\gamma_t,\sigma_t) \, := \ \qf_t(\opH) =  \qf_0(\opH)  \,,
\eeqn
i.e., $\mathcal{E}(\phi_t,\gamma_t,\sigma_t)$ is \textit{independent} of $t$; 
see Corollary~\ref{cor-partnumb-en-conserv} and
Theorem~\ref{prop-en-vfull-1}, or
Prop~\ref{prop:quasi-hamiltonian-structure}.

\item  {\it Positivity preservation property}:  If $\Gamma=\begin{pmatrix}
\gamma & \sigma\\
 \bar{\sigma} & 1+\bar \gamma
\end{pmatrix}\geq 0$ at $t=0$, then this holds for all times.

\item Global well-posedness
  (Theorem~\ref{thm:existence-uniquenes-up-to-Coulomb}) of the HFB
  equations.

 \end{itemize}

It is in the proof of the local existence part of the last statement (Lemma~\ref{prop:continuity_of_B_and_K}\eqref{enu:continuity-B-H1L2-LH1}) that an error was made in  \cite{BachBreteauxChenFroehlichSigal2016}. 
In Appendix \ref{sec:ops-b-k} (Lemma~\ref{prop:continuity_of_B_and_K}\eqref{enu:continuity-B-H1L2-LH1}) we prove the corresponding estimate under a more restrictive condition on the pair potential $v$ - Condition (b') above.

In \cite{BenedikterSokSolovej2018}, the program outlined in this paper has been pursued for equations analogous to the HFB equations valid for fermions, namely the Bogolubov-de Gennes equations; see also \cite{ChenSigal2017}. For references to related work see \cite{BachBreteauxChenFroehlichSigal2016, ChenSigal2017, BenedikterSokSolovej2018}.

We will show that any observable conserved by the von
  Neumann-Landau dynamics which is linear or quadratic in the
  creation- and annihilation operators is also conserved by the
  quasifree dynamics; see
  Theorem~\ref{thm:preservation-particle-number}. In the special 
  case of the observable $\opN$, this yields the statement above.
  Energy conservation follows from Eq. \eqref{eq-vNeum-quasifree}, with $\mathbb{A}= \opH_{\rm hfb}(\qf_t)$,
  the quadratic nature of $\opH_{\rm hfb}(\qf_t)$, and Eq. \eqref{eq-qf-consistency-0}.
  
Note that conservation of the  total number of atoms in the gas is a consequence of (global) $U(1)$-gauge invariance, 
i.e., invariance of  the Hamiltonian $\opH$ under the transformations 
$$\psi(x) \rightarrow e^{i\theta}\psi(x), \qquad \psi^{*}(x) \rightarrow e^{-i\theta} \psi^{*}(x), \qquad \forall \theta \in \mathbb{R}, \forall x \in \mathbb{R}^{d}.$$

  The  total particle number, $\cN (\phi,\gamma,\sigma) \, := \, \qf (\opN)$, and energy, $\mathcal{E} (\phi,\gamma,\sigma)  :=\qf (\opH)$, as functions of $ (\phi,\gamma,\sigma)$, can be evaluated explicitly:  
\eqn 
    \cN(\phi,\gamma,\sigma) \,  = \,  \int \big(\gamma (x;x)+|\phi (x)|^2\big) dx\, .
\eeqn 
The energy $\mathcal{E} (\phi,\gamma,\sigma)$ is given explicitly in Eq.
\eqref{eq:energy}. For a delta-function pair potential, $v=g\delta$, 
it takes the form
\begin{align} \label{eq-energy-delta-1}
    \mathcal{E}(\phi,\gamma,\sigma)
   & \; = \; \tr[h(\gamma+|\phi\ran\lan\phi|)]\\
    & 
     +g\int\big(2n(x)|\phi(x)|^{2}
    +n(x)^{2}+\frac{1}{2}|w(x)|^{2}\big)dx \,.
    \nonumber
\end{align}
(In terms of $\opH_{\rm hfb}(\qf)$, we have that $\mathcal{E}(\phi,\gamma,\sigma)  :=\qf(\opH) =\qf(\opH_{\rm hfb}(\qf))+\text{scalar}$.)

As usual, if $\gam$ is trace-class and $\s$ is Hilbert-Schmidt the energy functional $\mathcal{E}$ can be used to give a variational characterization of stationary Gibbs states:
\begin{itemize}\item Gibbs states minimize the energy $\mathcal{E}(\phi,\gamma,\sigma)$ under the constraint of constant entropy and for a fixed value of the expected particle number.
\end{itemize}

Equation~\eqref{eq-qf-consistency-0} suggests to define {\it HFB stationary states} as the quasifree states satisfying the equation
\eqn\label{sc-eq-stat}
\qf\big( [\opA,\opH_{\rm hfb}(\qf)] \big) \ = \ 0 \,,
\eeqn 
for all observables $\opA$. (If $\qf$ is given by a density matrix, we
can rewrite this equation as an explicit fixed point equation, see
\eqref{fixed-point-Gibbs-quasifree-state} below.) The most interesting ones
among such states are the \textit{ground states} and \textit{Gibbs states}.  These
states are defined as 
\[
\qf_{\beta, \mu} \ := \ \lim_{L\to \infty} \qf_L,
\]
where $\qf_L$ is the quasifree ground state or Gibbs state of a
Bose gas confined to a torus, $\Lambda_L=\R^d/2 L \Z^d$, i.e., to the box
$[- L, L]^d$ with periodic boundary conditions. It satisfies the
fixed point equation 
\eqn\label{fixed-point-Gibbs-quasifree-state}
  \Phi_{\beta,\mu}(\qf_L) \ = \ \qf_L\,, \,\,\,\text{with}\,\,\,
  \Phi_{\beta,\mu}(\qf_L)(\opA) \ := \
  \tr[\opA\,\exp(-\beta(\opH_{\text{hfb}}(\qf_L)-\mu\opN))]/\Xi \, , 
\eeqn 
where $\beta>0$ is the inverse temperature, $\mu$ is the
chemical potential, and
$\Xi=\tr[\exp(-\beta(\opH_{\text{hfb}}(\qf_L))-\mu\opN)]$ is the partition
function (the exponential of the negative pressure) of the gas. The quasifree state $\qf_L$
  fulfilling \eqref{fixed-point-Gibbs-quasifree-state} is a solution to Eq.~\eqref{sc-eq-stat} (or a stationary solution of Eq.~\eqref{eq-qf-consistency-0}) for a gas confined to the box
$\Lambda_L$. With regard to the thermodynamic limit, $L\rightarrow \infty$, we note that if the external potantial $V$
vanishes (i.e, for a translation-invariant Hamiltonian),
\[
\qf\big( [\opH_{\text{hfb}}(\qf), \opA] \big) 
\ := \ 
\lim_{L\to \infty} \qf_L\big( [\opH_{\text{hfb}}(\qf_L), \opA] \big)
\ = \ 0 ,
\] 
for any observable $\opA$ localized in a compact region of position space.

Furhtermore, if the external potential $V$ vanishes (the translation-invariant case) one should replace the total energy and the particle number by the energy density and particle density, respectively, in order to study the approach to the thermodynamic limit, 
 $L \rightarrow \infty$. 

 
If $V=0$ and $\g$ and $\s$ are translation-invariant, then the integrand in the energy functional $\mathcal{E}(\phi, \gamma, \sigma)$ (see \eqref{eq:energy}) is the energy density functional introduced in \cite{CritchleySolomon1976} and further studied in \cite{NapiorkoswkiReuversSolovej2018a, NapiorkoswkiReuversSolovej2018b}. 
  It is shown in the latter papers that this functional has minimizers under the constraint of constant entropy- and particle densities. 
In \cite{NapiorkoswkiReuversSolovej2018a, NapiorkoswkiReuversSolovej2018b} it is also shown that a condensate appears in the corresponding minimizers. (To complete the picture one should show that the states thus obtained are stationary solutions to equation~\eqref{eq-qf-consistency-0}.)


In this paper we do not consider the general problem of existence of  static solutions. However, for $V=0$,  we present a result concerning existence of the positive-temperature, $U(1)$-gauge- and translation-invariant HFB Gibbs states, and we show that Bose-Einstein condensation (BEC) occurs above a critical density; see Theorem~\ref{thm:existence-bose-condensate}. 

As mentioned above, for $U(1)$-gauge-invariant quasifree states, $\phi=0$ and $\sigma=0$; and hence HFB Gibbs states with these properties are, in fact, stationary solutions of the bosonic Hartree-Fock equation. Moreover, as the results of \cite{NapiorkoswkiReuversSolovej2018a, NapiorkoswkiReuversSolovej2018b} show, in the BEC regime, these states are not minimizers of the full HFB energy denisty, at fixed values of the entropy- and particle density.
However, the existence of such states exhibiting  Bose-Einstein condensation suggests that there are also $U(1)$-symmetry breaking HFB Gibbs states with  $\phi\neq 0$ and $\sigma\neq 0$.

\subsection{Fixed Point Equation}
Let $\UHFB(t,s)$ denote the unitary propagator on bosonic Fock space $\cF$, see \eqref{Fock},
solving
\begin{align}\label{eqn:evolution-td-quartic-hamiltonian}
    i\partial_{t}\UHFB(t,s) =\HHFB(\qf_t)\, \UHFB(t,s)\,, \qquad \text{ with }\, \UHFB(s,s) & =\1\,, \forall s\,.
\end{align} 
In terms of this propagator, we can rewrite equation \eqref{eq-qf-consistency-0}, with initial condition $\qf_0 =\omega_0$, as 
a fixed point problem, 
\begin{equation} \label{FPproblem}
\qf_{t}=\Phi_{t}(\qf_{(\cdot)}), \qquad \text{with  }\,\, \Phi_{t}(\qf_{(\cdot)})(\opA) :=\qf_0(\UHFB(t,0)^{*}\, \opA \, \UHFB(t,0)),
\end{equation}
for all times $t\in \mathbb{R}$. Since the propagators $\UHFB(t,s)$ are 
generated by families of quadratic Hamiltonians, we have that $\qf_0(\UHFB(t,0)^{*}\, \opA \, \UHFB(t,0))$
 is a quasifree state, for any time $t$. This formulation opens the possibility 
 to prove existence of the quasifree dynamics directly, using a Brouwer-Schauder-type fixed-point theorem, 
 without passing to the truncated expectations $\phi$, $\gamma$ and $\sigma$.

In this paper, we do not study whether the quasi-free effective dynamics  \eqref{eq-qf-consistency-0} (the HFB equations) provide an
accurate approximation to the many-body dynamics \eqref{eq-vNeum-1},
for finite times. There is a large literature concerning the
derivation of the simpler Hartree- and Hartree-Fock equations from many-body dynamics in a limiting (mean-field) regime.
Recently, evolution equations that include linear fluctuations around solutions of the
Hartree equation (i.e., equations arizing from linearization of the HFB equation in
$\gamma$ and $\sigma$) have been derived in
\cite{GrillakisMachedon2013a, LewinNamSchlein2015, NamNapiorkoswki2017, LewinNamRougerie2017, LewinNamRougerie2016};
see \cite{Lewin2015} for a recent review, and \cite{JackiwKerman1979} for an early
contribution. Very recently, it has been brought to our attention that, independently and in a different
framework, equations equivalent to
\eqref{eq:BHF-phi}--\eqref{eq:BHF-sigma} 
are derived for \textit{pure states} in some recent papers
\cite{GrillakisMachedon2013b,GrillakisMachedon2017}. 
For pure quasifree states, the relation $\gamma + \gamma^2 = \sigma \sigma^*$
holds, and equations
\eqref{eq:BHF-phi}--\eqref{eq:BHF-sigma} turn out to be hamiltonian evolution equations.
\subsection{Organization of the paper}
In Section~\ref{sec-HFB-general-1}, we first present the HFB equations, which we derive in Appendix~\ref{sec-HFB-deriv-1}.  We then show that certain conservation laws for the many-body problem imply corresponding conservation laws for the HFB equation.\\
In Section~\ref{sec:symplectic-hamiltonian-structure}, we show that the space of solutions of the HFB equations has a symplectic structure, and that these equations have similarities with Hamiltonian equations of motion. \\
In  Section~\ref{sec:Relation_with_HFB_eigengalues_equations}, we explain how the symplectic version of the HFB equations is related to the HFB eigenvalue equations found in the physics literature. \\
In Section~\ref{sec-GWP-HFB-1}, we prove that the Cauchy problem for the HFB equations is globally well-posed in the ``energy space'', provided that the pair interaction potential is assumed to have suitable regularity properties.
Our proof of global well-posedness is inspired in part by previous
work on the Hartree-Fock equation \cite{BoveDaPratoFano1974, ChadamGlassey1975, BoveDaPratoFano1976,
  Chadam1976, Zagatti1992}.  We note that global existence for the related
time-dependent Bogolubov-de Gennes equations for fermion systems has
recently been established in \cite{BenedikterSokSolovej2018}, using a
similar proof strategy.\\
In Section~\ref{sec-Gibbs-HFB-1}, we prove Bose-Einstein
condensation for stationary states. \\
 A brief summary of the theory of quasifree
states and proofs of various technical lemmata are
collected in Appendices.

\subsection*{Acknowledgements}
The work of S.B.\ is supported by the Basque Government through the
BERC 2014-2017 program, and by the Spanish Ministry of Economy and
Competitiveness MINECO (BCAM Severo Ochoa accreditation SEV-2013-0323,
MTM2014-53850), and the European Union's Horizon 2020 research and
innovation programme under the Marie Sklodowska-Curie grant agreement
No.~660021. The work of T.C.\ is supported by NSF grants DMS-1151414
(CAREER) and DMS-1716198. The work of I.M.S.\ is supported in part by
NSERC Grant No.~NA7901 and SwissMAP Grant.

\section{The HFB equations and their basic properties}
\label{sec-HFB-general-1}

In this section, we formulate the HFB equations for a general pair
potential $v$  and prove the associated conservation laws. The derivation of the HFB equations is done in Appendix \ref{sec-HFB-deriv-1} by applying the quasifree reduction as in the introduction.

\paragraph{{\bf Definition of spaces.}} 
Let $M:=\langle \nabla_{x}\rangle=\sqrt{1-\Delta_{x}}$, 
with $\Delta_{x}$ being the Laplacian in $d$ dimensions. 
We denote by $\mathcal L^1$ and $\mathcal L^2$  the spaces of  trace-class and Hilbert-Schmidt operators on  $L^2(\mathbb R ^d)$ endowed with the trace norms $\|\cdot\|_{\mathcal L^1}$ and $\|\cdot\|_{\mathcal L^2}$.  For $j\in \mathbb N_0$ we define the spaces  
\begin{align}\label{eq-Xjspace-def-1}
X^{j} & =\big\{(\phi,\gamma,\sigma)\in 
H^{j} \times \mathcal{H}^{j}_\g \times \cH^{j}_\s \big\}\,,
\end{align}
with $H^{j}$ being the Sobolev space $H^{j}(\mathbb{R}^{d})$,  
 $\mathcal{H}^{j}_\g:=M^{-j}\mathcal{L}^{1} M^{-j}$, 
and $\mathcal{H}^{j}_\s:=\{\s\in \mathcal L^2: \|M^{j}\s\|_{\mathcal L^2} + \|\s M^{j}\|_{\mathcal L^2} < \infty\}$ with the norms 
\begin{align}\label{norms}\|\phi\|_{H^j}:=\|M^{j}\phi\|_{L^{2}},\  \|\gamma\|_{\mathcal{H}^j_\g}:=\|M^{j}\gamma M^{j}\|_{\mathcal{L}^{1}},\  \|\s\|_{\mathcal{H}^j_\s}:=\|M^{j}\s \|_{\mathcal{L}^{2}} + \|\s M^{j}\|_{\mathcal{L}^{2}}.\end{align}
We endow the spaces $X^{j}$ with the norms \[\|(\phi,\gamma,\sigma)\|_{X^{j}}=\|\phi\|_{H^j} +  \|\gamma\|_{\mathcal{H}^j_\g}+ \|\s\|_{\mathcal{H}^j_\s}\] 

Furthermore, we let  $\cX_T:=C^{0}([0,T); X^{3})\cap C^{1}([0,T); X^{1})$ and 
 we denote by $X^{j}_{\rm qf}$ and $\cX_T^{\rm qf}$ the spaces of quasifree states and families of quasifree states with the $1^{st}$ and $2^{nd}$ order truncated expectations  from the spaces $X^{j}$ and $\cX_T$, respectively. 
 
\begin{remark}  For systems with infinite number of particles and finite density, one could replace $\R^d$ by the torus $\T^d_L:=\R^d/(L \Z)^d$ and then pass to the thermodynamic limit.
\end{remark} 

\DETAILS{We denote by $\mathcal B$ the space of bounded operators on 
$L^2(\mathbb R ^d)$, with the operator norm denoted by $\|\cdot\|$. 
For $j\in \mathbb N_0$ we define the spaces  
\begin{align}\label{eq-Xinfty-jspace-def-1}
X^{j, \infty} & =\big\{(\phi,\gamma,\sigma)\in 
H^{j} \times \cB^{j} \times \cB^{j}: \gamma=\gamma^*\ge 0\ \text{ and }\ \sigma^*=\bar\sigma
 \big\}\,,
\end{align}
with $H^{j}$ the Sobolev space $H^{j}(\mathbb{R}^{d})$ and $\cB^{j}=M^{-j}\mathcal B M^{-j}$. 
The norms on $X^{j, \infty}$ are given by the norms on the Banach spaces, 
$H^{j} \times \cB^{j} \times \cB^{j}$, i.e., with $\|\cdot\|_{L^{2})}\equiv \|\cdot\|_{L^{2}(\mathbb{R}^{d})}$, m
\[
\|(\phi,\gamma,\sigma)\|_{X^{j, \infty}}= \|M^{j}\phi\|_{L^{2}} + \|M^{j}\gamma M^{j}\| + \|M^{j}\sigma M^{j}\|.
\]
(The superindex $\infty$ indicates that we are dealing with bounded operators, as opposed to the trace-class and Hilbert-Schmidt ones appearing later.)

 Moreover, we let  $\cX_T^{\infty}:=C^{0}([0,T); X^{j, \infty})\cap C^{1}([0,T); X^{0, \infty})$,  for a fixed $j$ satisfying $j> d/2$ and $j\ge 2$, and   denote by $X^{j, \infty}_{\rm qf}$ and $\cX_{T, \rm qf}^{\infty}$ 
spaces of quasifree states and families of quasifree states with the $1^{st}$ and $2^{nd}$ order truncated expectations  from the spaces $X^{j, \infty}$ and $\cX_T^{\infty}$, respectively. 

The reason for introducing the spaces $\cB^{j}$ is the following elementary result.
\begin{lemma} \label{lem:d-k-prop}  Any $\al\in\cB^{j}, j>d/2,$ has a bounded, H\"older continuous integral kernel $\al(x, y)$.

Let $v\in L^1$. Then the operator $k : \al \ra v \,\sharp\, \al$, defined through its integral kernel:
\eqn
v \,\sharp\, \al \, (x;y):= v(x, y)\al (x;y) \,,
\eeqn
is bounded  from $\cB^{j}, j>d/2,$  to $\cB^{s}, s<  j-d/2$. \end{lemma}
\begin{proof} 
Let $G(x)$ be the Fourier transform of the function $(1+|\xi|^2)^{-j/2}, j>d/2$. Then $G\in H^{s}(\mathbb{R}^{d}), s<j-d/2,$ and we can write $\al(x, y)=\lan G_x, \al_1 G_y\ran,$
where $G_x(x'):=G(x-x')$ and $\al_1:=M^j \al M^j$. Since $\al\in\cB^{j}, j>d/2,$ $\al_1$ is a bounded operator and therefore $|\al(x, y)| \le \| G_x\|_{L^2} \|\al_1\| \| G_y\|_{L^2}=\|\al_1\| \| G\|_{L^2}^2$. Similarly, one shows the H\"older continuity.

For the second statement, the result above gives $\| (v \,\sharp\, \al) f\|_{L^2} \le \| \al(\cdot)\|_{L^\infty} \| |v| * |f|\|_{L^2}\le \| \al(\cdot)\|_{L^\infty} \| v\|_{L^1} \| f\|_{L^2}$. Similarly, one shows the estimates involving $M^s$. \end{proof}}


In what follows, we assume Conditions (a) and (b)
stated in the Introduction [see Eqs.~\eqref{cond-a} and \eqref{cond-b}].

\begin{theorem}\label{prop-HFB-vfull-1}
The family of quasifree states $\qf_t\in \cX_{T}^{\rm qf} 
$  satisfies 
\eqn\label{eq-omt-commut-2}
    i\partial_t\qf_t(\opA) = \qf_t([\opA,\opH])
    \;,\;\;
    \forall \; \opA  \in \mathcal{A}^{(2)} \,, \; 
\eeqn
with the Hamiltonian $\opH$ defined in \eqref{model-ham}, if and only if the triple $(\phi_t,\gamma_t,\sigma_t)\in \cX_T 
$ of the  $1^{st}$ and $2^{nd}$ order truncated expectations of $\qf_t$ satisfies the time-dependent
Hartree-Fock-Bogoliubov equations 
\begin{align}
i\partial_{t}\phi_{t} & =h (\gamma_t)\phi_{t}+k (\sigma^{\phi_{t}}_t)\bar{\phi}_{t} \label{eq:BHF-phi}\,,\\
i\partial_{t}\gamma_{t} & =[h (\gamma_t^{\phi_{t}}),\gamma_{t}]+k (\sigma^{\phi_{t}}_t)\sigma_{t}^*-\sigma_t k (\sigma^{\phi_{t}}_t)^* \label{eq:BHF-gamma}\,,\\
i\partial_{t}\sigma_{t} & =[h (\gamma_t^{\phi_{t}}),\sigma_{t}]_{+}+[k (\sigma^{\phi_{t}}_t),\gamma_{t}]_+ +k (\sigma^{\phi_{t}}_t), \label{eq:BHF-sigma}
\end{align}
where $[A_1,A_2]_{+}=A_1A_2^{T}+A_2A_1^{T}$, $\gamma^{\phi}:=\gamma+|\phi\rangle\langle\phi|$ and $\sigma^{\phi}:=\sigma +|\phi\rangle\langle\bar\phi|$, and 
\begin{align}\label{h}
& h (\gamma)  =h+b [\gamma]\,,\  b [\gamma]:= v*d(\gamma)  + v \,\sharp\, \gamma \,, \\
 \label{k}
 & k (\sigma)  =v\,\sharp \, \sigma\,,\ \quad d(\al)(x):=\al(x, x).
\end{align}
\end{theorem}
If $v=g\delta$, $h (\gamma)$ agrees with $h_{g \delta}(\gamma)$ in \eqref{eq-nmwh-notations-0},
and $k (\s)$ agrees with the   multiplication operator by $g\ d(\s)(x)$ in \eqref{eq-nmwh-notations-0}.

Due to 
the fact that $h(\gamma_t)$ is $\Delta-$bounded, for each $t>0$, the r.h.s. of \eqref{eq:BHF-phi} - \eqref{eq:BHF-sigma} belongs to the space $X^{0}$. 
 The proof of Theorem \ref{prop-HFB-vfull-1} is given in Appendix \ref{sec-HFB-deriv-1}.

%
We now show that equations \eqref{eq-vNeum-quasifree} (or \eqref{eq:BHF-phi} to~\eqref{eq:BHF-sigma}) and \eqref{eq-qf-consistency-0} describing the quasifree dynamics are equivalent.

For a quasifree state $\qf$ with $1^{st}$ and $2^{nd}$ order truncated expectations $(\phi,\gamma,\sigma)\in \XL$, we define the quadratic Hamiltonian parametrized by $(\phi,\gamma,\sigma)$ as 
\begin{align}\label{eq-HHFB-def-1}
    \HHFB(\qf) 
    & =\int \psi^{*}(x)h_v(\gamma)\psi(x)\, dx 
    \nonumber \\
    & \quad -\int \psi^{*}(x) b [|\phi \rangle\langle\phi|]\phi (x)\, dx+h.c.  
    \nonumber \\
    & \quad+\frac{1}{2}\int\psi^{*}(x)k (\sigma)\psi^{*}(x)\, dx+h.c. \,.
\end{align}

\begin{theorem}
\label{thm-Hquadr-dyn-1}
%
Equation \eqref{eq-vNeum-quasifree} is 
 equivalent to the nonlinear, self-consistent evolution equation 
\eqn\label{eq-omt-commut-3}
    i\partial_t \qf_t(\opA) = \qf_t ({[\opA, \HHFB(\qf_t)]}) \,,
\eeqn
defined for all observables $\opA$. The equivalence holds for observables linear or quadratic in creation- and annihilation operators.
 
 Moreover, runcated expectations  $(\phi_t,\gamma_t,\sigma_t)\in \cX_T$ 
 satisfy the HFB equations \eqref{eq:BHF-phi} to~\eqref{eq:BHF-sigma} if and only if 
the corresponding  quasifree state $\qf_t\in \cX_{T}^{ \rm qf}$ satisfies \eqref{eq-omt-commut-3}.\end{theorem}
The proof of Theorem \ref{thm-Hquadr-dyn-1} is given in Appendix \ref{sec-HFB-equiv-1}.

We now prove the conservation laws for the number of particles (or more generally, for any observable commuting with the Hamiltonian $\mathbb H$ which is quadratic
with respect to creation and annihilation operators), and for the energy.

\begin{theorem}
\label{thm:preservation-particle-number}
Assume that an observable $\opA\in \cA^{(2)}$ satisfies  $[\opH,\opA]=0$. 
 Then $\qf_t(\opA)$ is   conserved:
 \eqn 
    \qf_t(\opA) = \qf_0(\opA) \;\;\;\;\forall\;t\in\R\,.
\eeqn
\end{theorem}
\prf
This follows from \eqref{eq-omt-commut-2} for $\opA$ of order up to two, with $[\opA,\opH]=0$. 
\endprf
To draw some consequences from this result we need to define additional spaces.



\begin{corollary} \label{cor-partnumb-en-conserv} 
  Let $\qf_t\in \cX_T^{\rm qf}$ solve \eqref{eq-omt-commut-2} (or \eqref{eq-omt-commut-3}). Then the number of particles $\cN(\phi_t,\gamma_t,\sigma_t)=\qf_t(\opN)$ and the energy $\qf_t(\opH)$ are conserved. 
 \end{corollary}
 
\begin{theorem} \label{prop-en-vfull-1} 
Let 
 $\qf\in X^{\rm qf}$.
Then  the energy $\qf(\opH) = \mathcal{E}(\phi,\gamma,\sigma)$ is given explicitly as 
\begin{multline}
    \mathcal{E}(\phi,\gamma,\sigma)=\tr[h(\gamma+|\phi\rangle\langle\phi|)]
    +\tr[b[|\phi\rangle\langle\phi|]\gamma]\\
    +\frac{1}{2}\tr[b[\gamma]\gamma]+\frac{1}{2}\int v(x-y)|\sigma(x,y)
    +\phi(x)\phi(y)|^{2}dxdy\,.
    \label{eq:energy}
\end{multline}  
\end{theorem}
\prf
We use 
\eqn
    \om_{C}^{q}(\mathbb{A}) :=\om^{q}(W_{\phi}\opA W_{\phi}^{*}) \,,
\eeqn
where the Weyl operators are defined through 
$W_{\phi}=\exp\big(\psi^{*}(\phi)-\psi(\phi)\big)$ and satisfy 
\begin{align}
    W_{\phi}^*\psi(x)W_{\phi}=\psi(x)+\phi(x) \, .
\end{align}
Note that the state $\om_{Ct}^{q}$ is quasifree because $\om^{q}$
is quasifree. By construction $\om_{C}^{q}(\psi(x))=0$ and thus using \eqref{eqn:quasi-wick} and the quasifreeness of $\om_{C}^{q}$ one sees that 
$\om_{C}^{q}$ vanishes on monomials of odd order in the creation and annihilation operators. 
Note that $\mathcal{E} (\phi,\gamma,\sigma)
 = \om_{C}^{q}(W_{\phi}^*\opH W_{\phi})$, hence using the vanishing on monomials of odd order in the creation and annihilation operators
\begin{align*}
\mathcal{E} (\phi,\gamma,\sigma)
& =
\om_{C,t}^{q}\Big(\int  v(x-y)\psi^*(x)\psi^*(y)\psi(x)\psi(y)dxdy\\
 & \quad + \frac{1}{2}\big(\int  v(x-y)\phi_t(x)\phi_t(y)\psi^*(x)\psi^*(y)dxdy +h.c.\big)\\
& \quad +\int \big(h+b[|\phi \rangle\langle\phi_t|]\big)(x;y) \psi^*(x)\psi(y)dxdy \Big)\\
& \quad + \frac{1}{2}\int |\phi (x)\phi (y)|^{2}v(x-y)dxdy +\langle \phi, h\phi \rangle \,.
\end{align*}
Then, using that $\om_{C}^{q}$ is a quasifree state with expectations $(0,\gamma,\sigma)$ yields
\begin{align*}
\mathcal{E}  (\phi ,\gamma,\sigma)
& = \frac{1}{2}\tr[b[\gamma]\gamma]+\frac{1}{2}\int\overline{\sigma (x,y)}v(x-y)\sigma (x,y)dxdy)\\
& \quad + \Re\Big(\int\overline{\sigma (x,y)}v(x-y)\phi (x)\phi (y)dxdy\Big)\\
& \quad +\tr[(h+b[|\phi \rangle\langle\phi|])\gamma]
+ \frac{1}{2}\int |\phi (x)\phi (y)|^{2}v(x-y)dxdy +\langle \phi, h\phi\rangle
\end{align*}
which gives the expression of the energy in terms of $\phi $, $\gamma $ and $\sigma $.
\endprf

\section{Generalized One-particle Density Matrix and Bogolubov Transforms}
\label{sec:symplectic-hamiltonian-structure}

In this section, we consider  the HFB equations \eqref{eq:BHF-gamma} - \eqref{eq:BHF-sigma} for $\gamma_t$ and $\sigma_t$ and reformulate them in terms the  generalized one-particle density matrix 
$\Gamma_t=\big(\begin{smallmatrix}\gamma_t & \sigma_t\\
\bar{\sigma_t} & 1+\bar{\gamma_t} \end{smallmatrix}\big)$. 
We show that the  diagonalizing maps for $\Gamma_t$ are symplectomorphisms (see below for the definition) and that the resulting equation for $\Gamma_t$ is equivalent to the evolution equation for these symplectomorphisms. The latter will allow us to (a)  give another proof of the conservation of  energy without using to the second quantization framework and (b) connect  the time-dependent HFB equations \eqref{eq:BHF-gamma} - \eqref{eq:BHF-sigma} to the  time-independent HFB equations used in the physics literature. See Section~\ref{sec:Relation_with_HFB_eigengalues_equations}.

We begin by relating properties of $\Gamma=\big(\begin{smallmatrix}\gamma & \sigma\\
\bar{\sigma} & 1+\bar{\gamma}
\end{smallmatrix}\big)$ to those of $\gamma$ and $\sigma$.

\begin{proposition} \label{prp-gampos-sigsym-1}
The generalized one-particle density matrix, $\Gamma$, satisfies:
\begin{align}\label{eqn:positivity-Gamma}
\Gamma=\begin{pmatrix}
\gamma & \sigma\\
 \bar{\sigma} & 1+\bar \gamma
\end{pmatrix}\geq 0 \,.
\end{align}
This property is equivalent to the following statements:
\begin{enumerate}
\item
The operator $\gamma\geq0$ is positive semidefinite.
\item
The expectation $\sigma(x,y)=\sigma(y,x)$ is symmetric.
\item \label{point:operator-inequality-gamma-sigma}
The inequality $\sigma (1+\bar \gamma)^{-1} \sigma^* \leq \gamma $ holds true in the sense of quadratic forms.
\item
The bound 
$\frac{1}{2}\|\sigma\|^2_{\Hsdd}\leq \|\gamma\|_{\cHLd}(1+ \tr[\gamma])$ 
holds true.
\end{enumerate}
 (Statement (4) follows from (1) and (3) and is given here for later convenience of references.) \end{proposition}
\prf
We remark that the truncated expectations $\gamma$ and $\sigma$ are the expectations of the state 
\[    \om_{C}(\mathbb{A}) :=\om(W_{\phi} \opA W_{\phi}^{*})\]
 where $W_{\phi}=\exp\big(\psi^{*}(\phi)-\psi(\phi)\big)$ are the Weyl operators. $W_{\phi}$ satisfy 
$W_{\phi}\psi(x)W_{\phi}^*=\psi(x)-\phi(x)$. 
The generalized one particle density matrix $\Gamma$ of $\om_{C}$ is non-negative, since, for all $f,g$ in $L^2$,
\eqn\label{Gam-posit}
\Big\langle\begin{pmatrix}f\\
g\end{pmatrix},\begin{pmatrix}\gamma & \sigma\\
\bar{\sigma} & 1+\bar{\gamma}\end{pmatrix}\begin{pmatrix}f\\
g\end{pmatrix}\Big\rangle=\om_{C}\big((\psi^{*}(f)+\psi(\bar{g}))(\psi(f)+\psi^{*}(\bar{g}))\big) \geq 0\,.
\eeqn

Statements (1) and (2) are obvious. The inequality in Point \eqref{point:operator-inequality-gamma-sigma} follows from the Schur complement argument:
\begin{align*} 
0\leq \begin{pmatrix}1 & -\sigma(1+\bar{\gamma})^{-1}\\
0 & 1
\end{pmatrix}
&\begin{pmatrix}\gamma & \sigma\\
\sigma^{*} & 1+\bar{\gamma}
\end{pmatrix}
\begin{pmatrix}1 & -\sigma(1+\bar{\gamma})^{-1}\\
0 & 1
\end{pmatrix}^{*} \\
=&\begin{pmatrix}\gamma-\sigma(1+\bar{\gamma})^{-1}\sigma^{*} & 0\\
0 & 1 + \bar \gamma
\end{pmatrix}\,.
\end{align*} 

Finally, we observe that (1) and (3) and the inequality $\g \le \tr[\gamma] \1$ imply 
  the following bound on $\sigma\sigma^*$, 
$$(1+\tr[\gamma])^{-1}\sigma\sigma^*\leq \sigma (1+\bar{\gamma})^{-1}\sigma^* \leq \gamma.$$ 
Inserting  $M=\sqrt{1-\Delta_x}$ on both sides and taking the trace yields (4). 
\endprf

\paragraph{{\bf Notations.}} 
With the spaces and norms defined after \eqref{eq-Xjspace-def-1} and 
for $j\in \mathbb N_0$ we define the spaces  
\begin{align}\label{Yjspace-def}
Y^{j} & =\mathcal{H}^j_\g \times \mathcal{H}^j_\s,
\end{align}
with the norms on $Y^{j}$ are given by 
\[
\|(\gamma,\sigma)\|_{Y^{j}}=\|\gamma\|_{\mathcal{H}^j}+\|\sigma\|_{H^j_s}.
\]
We also use the spaces $\cY_T:=C^{0}([0,T); Y^{3})\cap C^{1}([0,T); Y^{1})$ and $\tilde \cY_T:=$ the space of  generalized one-particle density matrices, $\Gamma$, with entries in $\cY_T$.

In what follows we fix a number $T>0$ and a family $\phi_t\in C^{0}([0,T);H^{3})\cap C^{1}([0,T); H^1)$ (not necessarily a solution \eqref{eq:BHF-phi}) and do not display it in our notation. 
 A simple computation yields the first result of this section: 
\begin{proposition}\label{prop:GamEvol} 
 $(\gamma_t,\sigma_t)\in \cY_T$ is a solution to the HFB equations \eqref{eq:BHF-gamma} - \eqref{eq:BHF-sigma}  
iff $\Gamma_t=\big(\begin{smallmatrix}\gamma_t & \sigma_t\\
\bar{\sigma_t} & 1+\bar{\gamma_t} \end{smallmatrix}\big)\in \tilde \cY_T$  solves the equation 
\eqn\label{eq:Gam-eq}
i\partial_t\Gamma_t = \mathcal S \Lambda(\Gamma_t) \Gamma_t -  \Gamma_t  \Lambda(\Gamma_t) \mathcal S\,,
\eeqn
with $\Lambda(\Gamma) =  \big(\begin{smallmatrix} 
h(\g^\phi) & k(\s^\phi) \\ \overline{k(\s^\phi)} & \overline{h(\g^\phi)}
\end{smallmatrix}\big)$, 
where, recall, $h (\g)$ and $k (\s)$ are defined in \eqref{h} and \eqref{k}, and $\mathcal S=\big(\begin{smallmatrix} 1 & 0\\ 0 & -1\end{smallmatrix}\big)$. \end{proposition}
To formulate the next result we introduce some definitions.
\begin{definition}\label{def:symplecto-implementable}
Let $\mathfrak h$ denote a complex Hilbert space. 
A bounded linear operator $\mathcal U= \big(\begin{smallmatrix} u & v\\ \bar v & \bar u \end{smallmatrix}\big)$ on $\mathfrak h \times \mathfrak h$ with the property that 
\eqn \label{eq:def-symplecto}
\cU^* \cS \cU=\cS \qquad \text{and} \qquad \cU\cS\cU^* =\cS \,,
\eeqn
with $\mathcal S=\big(\begin{smallmatrix} 1 & 0\\ 0 & -1\end{smallmatrix}\big)$, is called a symplectomorphism.

If, moreover, there exists a unitary transformation $\mathbb U$ on the Fock space, sometimes called implementation of $\mathcal U$, such that
\[
\forall f,g\in\mathfrak h\,,\quad 
\mathbb U [\psi^* (f ) + \psi(\bar g)] \mathbb U^* = \psi^* (uf + vg) + \psi(v \bar f + u \bar g)\,,
\]
 then the symplectomorphism $\mathcal U$ is said to be implementable.
\end{definition}
\begin{remark}
The operator $\cU$ is a symplectomorphism in the sense that it preserves the symplectic form $\im\langle \, \cdot \, ,\cS  \, \cdot  \, \rangle$ on $\mathfrak h \times \mathfrak h$ (i.e. is a canonical map). (In fact, $\cU$ preserves $\langle \, \cdot \, ,\cS  \, \cdot  \, \rangle$.)
\end{remark}
\begin{remark}
The operator $\mathcal U$ is a symplectomorphism if and only if the operator $f\mapsto uf+v\bar f$ is a symplectomorphism on $(\mathfrak h,\im\langle\cdot,\cdot\rangle)$ in the usual sense (i.e., it preserves the symplectic form $\im\langle\cdot,\cdot\rangle$)
\end{remark}
\begin{remark}
The conditions in \eqref{eq:def-symplecto} are equivalent to satisfying the four equations
\begin{align}\label{eq:def-symplecto-alternative}
uu^*-vv^*=\1\,,\quad u^*u-v^T\bar v = \1 \,, \quad u^*v=v^T\bar u\,, \quad uv^T=vu^T \,.
\end{align}
\end{remark}
\begin{remark}
The transformation 
\begin{align}\label{Bog-transf}\forall f,g\in\mathfrak h\,,\quad 
(\psi^* (f ), \psi(\bar f)) \ra ( \psi^* (uf) + \psi(v \bar f), \psi^* (v f) + \psi(u \bar f))
\end{align}
is called the Bogoliubov transformation. It is easy to check that it preserves the CCR iff the  operator $\mathcal U= \big(\begin{smallmatrix} u & v\\ \bar v & \bar u \end{smallmatrix}\big)$ satisfies \eqref{eq:def-symplecto}.
\end{remark}

If $v$ is Hilbert-Schmidt, then  the Bogoliubov transformation \eqref{Bog-transf}  is implementable.  
This condition is referred to as the Shale condition; see \cite{Shale1962}.

For later use, we introduce the Banach space
\[
\SobSymp =\Big\{\begin{pmatrix}a & b\\
\bar{b} & \bar{a}
\end{pmatrix}\,\Big\vert\, 
a\in\mathcal{B}(\Hd)\simeq M\cB M^{-1}\,,\quad 
b\in M\cL^2 M^{-1}\Big\}\,,
\]
endowed with the norm $\big\|\big(\begin{smallmatrix}a & b\\
\bar{b} & \bar{a}
\end{smallmatrix}\big)\big\|_{\SobSymp}=\|a\|_{\mathcal B(\Hd)}+\|b\|_{M\cL^2 M^{-1}}$, using the same identification between operators and kernels as before.

We begin with an auxiliary result: 

\begin{proposition}
\label{thm:diagonalization-with-symplectomorphism} Let $\Gamma=\big(\begin{smallmatrix}\gamma & \sigma\\ \bar{\sigma} & 1+\bar{\gamma} \end{smallmatrix}\big)\in Y^1$  and $\Gamma\geq0$. 
 Then there exist an implementable symplectomorphism
$\mathcal{U}\in\SobSymp$ 
 such that
\[
\Gamma=\mathcal{U}\begin{pmatrix}\gamma' & 0\\
0 & 1+\overline{\gamma'}
\end{pmatrix}\mathcal{U}^{*}\,,
\]
where $0\leq\gamma'\leq\gamma$. The operator $\gamma'$ is unique up to conjugation by a unitary
operator.
\end{proposition}
This result  is related to Theorem~1 of \cite{NamNapiorkoswkiSolovej2016}, which is stronger. See also \cite{BachBru2016,Berezin1966}. As the relation between the two results is not obvious, we give a direct proof of Proposition \ref{thm:diagonalization-with-symplectomorphism} after the proof of Proposition~\ref{prop:symplectic-structure}.

The next result relates the evolution of $\Gamma_t$ to the evolution of  implementable symplectomorphisms $\mathcal U_t\in\mathcal H^{\infty,2}(\mathfrak{h} \times \mathfrak{h})$, diagonalizing  $\Gamma_t$.
\begin{proposition}\label{prop:symplectic-structure}
(i)   
For  any  $\Gamma_t\in \tilde \cY_T$ and any implementable symplectomorphism $\mathcal{U}_{0}\in\SobSymp$, the initial value problem
\eqn\label{eq:evolution-symplecto}
i\partial_t\mathcal U_t^* = \mathcal S \Lambda(\Gamma_t)\mathcal U_t^* \,,\quad
\mathcal U_{t=0}=\mathcal U_0\,,
\eeqn
has a unique solution in $\SobSymp$, which is a symplectomorphism for every $t$.

(ii)  Let $\Gamma_t\in \tilde\cY_T$ solve the equation \eqref{eq:Gam-eq}, with an initial condition $\Gamma_0\in \tilde Y^{3}$, s.t. $\Gamma_0 \ge 0$.
Let  $\mathcal U_0$ be an implementable symplectomorphism diagonalizing $\Gamma_0$: 
\[
\Gamma_0=\mathcal{U}_{0}\Gamma_0^\prime \mathcal{U}_{0}^{*} \,,\quad \Gamma_0^\prime =\begin{pmatrix}\gamma_{0}^\prime & 0\\ 0 & 1+\overline{\gamma_{0}^\prime}\end{pmatrix}\,.
\] 
Then the continuous family of implementable symplectomorphisms $\mathcal U_t$ in $\mathcal H^{\infty,2}(\mathfrak{h} \times \mathfrak{h})$  satisfying \eqref{eq:evolution-symplecto}, with the above  $\mathcal U_0$, diagonalizes  $\Gamma_t$: 
\begin{align}\label{eq:symplectic-transfo-Gamma}
\Gamma_t=\mathcal U_t^* \Gamma_0^\prime \mathcal U_t \geq 0\,.
\end{align}
\end{proposition}
\begin{proof}[Proof of Prop.~\ref{prop:symplectic-structure}]  
The operator $\Lambda_t$ can be decomposed as $\Lambda_t = \Lambda_{1} + \Lambda_{2,t}$ with
\[\Lambda_1 =
\begin{pmatrix}
h & 0\\
0 & \bar h
\end{pmatrix}
\,,\quad
\Lambda_{2,t} =
\begin{pmatrix}
b[\gamma_t+|\phi_t\rangle\langle\phi_t|] & k[\sigma_t+\phi_t\otimes\phi_t]\\
 \overline{k[\sigma_t+\phi_t\otimes\phi_t]} & \overline{b[\gamma_t+|\phi_t\rangle\langle\phi_t|]}
\end{pmatrix}\,.\]
The first operator, $\Lambda_1$, is the generator of a continuous one-parameter group in $\SobSymp$. As for the second one, using the continuity of $t\mapsto\rho_t\in \XL$, and  Lemma~\ref{prop:continuity_of_B_and_K}, we get the continuity of $t\mapsto\Lambda_{2,t}\in \SobSymp$. We can thus use classical results of functional analysis (see, e.g., \cite{Kato1970}) to obtain the existence and uniqueness of $\mathcal U_t$ and its regularity. 

The same arguments as in the next lemma prove that $\cU_t$ is a symplectomorphism.

Finally, $\Gamma_t$ and $\mathcal U_t^* \Gamma_0 \mathcal U_t $ satisfy the same differential equation, and the uniqueness of a solution to~\eqref{eq:evolution-symplecto} proves the last equality.
\end{proof}

\begin{proof}[Proof of existence in  Prop.~\ref{thm:diagonalization-with-symplectomorphism}] 
We split the proof into two lemmas, Lemmas~\ref{lem:construction-U_t-diagonalisation}
and \ref{lem:ordinary-differential-equation-giving-gamma-infty} below.
The strategy is to construct $\Gamma_{t}$ and symplectomorphisms
$\mathcal{U}_{t}$ such that $\mathcal{U}_{t}\Gamma_{t}\mathcal{U}_{t}^{*}=\Gamma_{0}$,
for all $t$, and in the limit $t\to\infty$, $\Gamma_{\infty}$ has
the desired form. The key step will be to use a differential equation
for $\Gamma_{t}$ implying $\|\sigma_{t}\|_{\Hsdd}\searrow0$.
\begin{lemma} \label{lem:construction-U_t-diagonalisation} Let $T>0$ and $\Lambda_{t}=\big(\begin{smallmatrix}a_{t} & b_{t}\\
\bar{b}_{t} & \bar{a}_{t}
\end{smallmatrix}\big)\in C([0,T);\SobSymp )$\textup{.}
Then, the ordinary differential equation
\begin{align}\label{eqn:symplectic-equation-for-U}
i\partial_{t}\mathcal{U}_{t}^{*}=\mathcal{S}\Lambda_{t}\mathcal{U}_{t}^{*}\,,
\end{align}
with initial data $\mathcal{U}_{0}^{*}=\big(\begin{smallmatrix}1 & 0\\
0 & 1
\end{smallmatrix}\big)$, 
 has a unique global solution $\mathcal{U}_{t}\in C^{1}([0,T);\SobSymp)$,
and $\mathcal{U}_{t}$ is a symplectomorphism for all time.

Moreover, if $\gamma_{t}\in C^1([0,T);\cHLd),\sigma_{t}\in C^1([0,T);\Hsdd)$ satisfy
\begin{align}
i\partial_{t}\gamma_{t} & =a_{t}\gamma_{t}-b_{t}\bar{\sigma_{t}}-\gamma_{t}a_{t}+\sigma_{t}\bar{b}_{t}\,,\label{eq:d-gamma-symplecto-general}\\
i\partial_{t}\sigma_{t} & =a_{t}\sigma_{t}-b_{t}(1+\bar{\gamma}_{t})-\gamma_{t}b_{t}+\sigma_{t}\bar{a}_{t}\,,\label{eq:d-sigma-symplecto-general}
\end{align}
with initial data $\sigma_{0}=\sigma$, $\gamma_{0}=\gamma$ given
in Prop~\ref{prop:symplectic-structure}(i), 
then, for all time $t$,
\eqn \label{Gam-evol}
\mathcal{U}_{t}\Gamma_{t}\mathcal{U}_{t}^{*}=\Gamma_{0}\,.
\eeqn
\end{lemma}
\begin{proof}   
The existence and uniqueness of $\mathcal{U}_{t}^{*}$ follows from
the theory of time-dependent linear ordinary differential equations once one observes that $ \cHLd $ and $\Hsdd$ are continuously embedded in $\cB(\Hd)$ and $M\cL^2(L^2)M^{-1}$.
At $t=0$, $\cU_{0}\cS\cU_{0}^{*}=\mathcal{S}$
and
\[
i\partial_{t}(\cU_{t}\cS\cU_{t}^{*})
=\cU_{t}\big(-\Lambda_{t}\cS\cS
+\cS\cS\Lambda_{t}\big)\cU_{t}^{*}=0\,,
\]
thus $\cU_{t}\cS\cU_{t}^{*}=\cS$
for all time, and, to prove $\cU_{t}^{*}\mathcal{S}\cU_{t}=\mathcal{S}$, one observes that
\[
i\partial_{t}(\cU_{t}^{*}\cS\cU_{t})
=-(\cU_{t}^{*}\cS\cU_{t}) \Lambda_{t}\cS +\cS\Lambda_{t}(\cU_{t}^{*}\cS\cU_{t})\,,
\]
which is a linear time-dependent ordinary differential equation for $\cU_{t}^{*}\cS\cU_{t}$ which also admits the constant solution $\cS$. By uniqueness of the solution, one gets that $\cU_{t}^{*}\mathcal{S}\cU_{t}=\mathcal{S}$. 
Hence $\cU_t$ is a symplectomorphism for all time.

Similarly, the derivative $i\partial_{t}\big(\mathcal{U}_{t}\Gamma_{t}\mathcal{U}_{t}^{*}\big)$
vanishes because, using (\ref{eq:d-gamma-symplecto-general}) and
(\ref{eq:d-sigma-symplecto-general}), 
\[
i\partial_{t}\Gamma_{t}=\Lambda_{t}\mathcal{S}\Gamma_{t}-\Gamma_{t}\mathcal{S}\Lambda_{t}\,.
\]
Thus $\mathcal{U}_{t}\Gamma_{t}\mathcal{U}_{t}^{*}=\mathcal{U}_{0}\Gamma_{0}\mathcal{U}_{0}^{*}=\Gamma_{0}$
for all times.
\end{proof}

We choose $a_{t}$ and $b_{t}$ in (\ref{eq:d-gamma-symplecto-general})
and (\ref{eq:d-sigma-symplecto-general}) such that $\sigma_{t}$ vanishes
in the limit $t\to\infty$. Let $\mathcal{L}^{1}(\mathfrak{h})$ and $\mathcal{L}^{2}(\mathfrak{h})$ denote the spaces of  trace-class and Hilbert - Schmidt operators on  the space $\mathfrak{h}$.
\begin{lemma}
\label{lem:ordinary-differential-equation-giving-gamma-infty}The
ordinary differential equation 
\begin{align}
\partial_{t}\gamma_{t} & =-2\sigma_{t}\bar{\sigma_{t}}\,,\label{eq:d-gamma-diagonalization}\\
\partial_{t}\sigma_{t} & =-(\sigma_{t}+\sigma_{t}\bar{\gamma}_{t}+\gamma_{t}\sigma_{t})\,,\label{eq:d-sigma-diagonalization}
\end{align}
with initial data $\sigma_{0}=\sigma$, $\gamma_{0}=\gamma$ given
in Prop.~\ref{prop:symplectic-structure}(i), 
 has a unique global solution $(\gamma_{t},\sigma_{t})\in C^{1}\big([0,\infty);\mathcal{L}^{1}(\mathfrak{h})\times\mathcal{L}^{2}(\mathfrak{h})\big)$.

Let $\Lambda_{t}=\big(\begin{smallmatrix}0 & i\sigma_{t}\\
-i\bar{\sigma}_{t} & 0
\end{smallmatrix}\big)$, and  $\mathcal{U}_{t}=\big(\begin{smallmatrix}
u_t & v_t\\ \bar v_t & \bar u_t
\end{smallmatrix}\big)$ and~$\Gamma_{t}=\big(\begin{smallmatrix}\gamma_t & \sigma_t\\
\bar{\sigma}_t & 1+\bar{\gamma}_t
\end{smallmatrix}\big)$ as in
Lemma~\ref{lem:construction-U_t-diagonalisation}:
\begin{itemize}
\item $\mathcal{U}_{t}$ converges in $\SobSymp$
to a symplectomorphism $\mathcal{U}_{\infty}$.
\item $\Gamma_{0}=\mathcal{U}_{\infty}\Gamma_{\infty}\mathcal{U}_{\infty}^{*}=\mathcal{U}_{\infty}\big(\begin{smallmatrix}\gamma_{\infty} & 0\\
0 & 1+\bar{\gamma}_{\infty}
\end{smallmatrix}\big)\mathcal{U}_{\infty}^{*}$ with $0\leq\gamma_{\infty}\leq\gamma_{0}$.
\end{itemize}
\end{lemma}
\begin{proof}  
The existence of maximal solutions to (\ref{eq:d-gamma-diagonalization})
- (\ref{eq:d-sigma-diagonalization}) follows from the  
Picard-Lindel\"of theorem. Now using the $\mathcal{U}_{t}$ constructed in Lemma~(\ref{lem:construction-U_t-diagonalisation}),
one gets that $(\mathcal{U}_{t})^{-1}\Gamma_{0}(\mathcal{U}_{t}^{*})^{-1}=\Gamma_{t}$,
which implies that $\Gamma_{t}\geq0$ and thus $\gamma_{t}\geq0$.
It then follows from (\ref{eq:d-gamma-diagonalization}) that $\gamma_{t}$
is decreasing in the sense of quadratic forms and $\|\gamma_{t}\|_{\cHLd}\leq\|\gamma_{0}\|_{\cHLd}$.

One first obtains an estimate on $\|\sigma_{t}\|_{\mathcal L^2}^{2}=\tr[\sigma_{t}\sigma_{t}^*]$, using \eqref{eq:d-sigma-diagonalization}:
\begin{align*}
\partial_{t}\|\sigma_{t}\|_{\mathcal L^2}^{2}
\ = \ &
\tr\big[ 
-(\sigma_{t}+\sigma_{t}\bar{\gamma}_{t}+\gamma_{
}\sigma_{t}) \sigma_{t}^{*} 
-\sigma_{t}(\sigma_{t}^{*}+\bar{\gamma}_{t}\sigma_{t}^{*}+\sigma_{t}^{*}\gamma_{t}) \big]
\\[1ex]
\ \leq \ & 
-2 \tr\big[ \sigma_{t} \sigma_{t}^{*} \big]
\ = \ 
-2\|\sigma_{t}\|_{\mathcal L^2}^{2}\,.
\end{align*}
This implies that $\|\sigma_{t}\|_{\mathcal L^2}\leq\|\sigma_{0}\|_{\mathcal L^2}\exp(-t)$.
Using again \eqref{eq:d-sigma-diagonalization} and the fact that $\gamma_{t}\geq0$ one finds that
\begin{align*}
\partial_{t}\|\sigma_{t}\|_{\Hsdd}^{2} &
 = 
\tr\big[ 
-(\sigma_{t} +\sigma_{t}\bar{\gamma}_{t}
+\gamma_{t}\sigma_{t}) \sigma_{t}^{*}  M^2
-\sigma_{t} (\sigma_{t}^{*} +\bar{\gamma}_{t}\sigma_{t}^{*} +\sigma_{t}^{*}\gamma_{t}) M^2
 \big]
\\
 & \leq-2\|\sigma_{t}\|_{\Hsdd}^{2}-\tr[\gamma_{t}\sigma_{t}\sigma_{t}^{*}M^{2}]-\tr[\sigma_{t}\sigma_{t}^{*}\gamma_{t}M^{2}]
\end{align*}
We remark that $|\tr[M\gamma_{t}\sigma_{t}\sigma_{t}^{*}M]|\leq\|\gamma_{t}\|_{\cHLd}^{1/2}\,\|\gamma_{t}^{1/2}\sigma_{t}\|_{{\mathcal B}(\mathfrak{h})}\,\|\sigma_{t}\|_{\Hsdd}$
and 
\[
\|\gamma_{t}^{1/2}\sigma_{t}\|_{{\mathcal B}(\mathfrak{h})}
\leq\|\gamma_{t}^{1/2}\|_{{\mathcal L}^{2}}\|\sigma_{t}\|_{{\mathcal L}^{2}}
\leq\|\gamma_{0}\|_{\tr}^{1/2}\|\sigma_{0}\|_{\mathcal L^2}e^{-t}\,,
\]
hence
\[
\partial_{t}\|\sigma_{t}\|_{\Hsdd}^{2}\leq-2\|\sigma_{t}\|_{\Hsdd}^{2}+\sqrt{2C}e^{-t}\,\sqrt{2}\|\sigma_{t}\|_{\Hsdd}\leq-\|\sigma_{t}\|_{\Hsdd}^{2}+Ce^{-2t}
\]
which yields $\|\sigma_{t}\|_{\Hsdd}^{2}\leq Ce^{-t}\|\sigma_{0}\|_{\Hsdd}^{2}$. 
The pair $(\gamma_{t},\sigma_{t})$ is thus bounded in $\cHLd \times \Hsdd$
and the maximal time of the solution is~$T=\infty$. We also get
that $\gamma_{t}\to\gamma_{\infty}$ in $\cHLd$
as $t\to\infty$ as $\gamma_{t}$ is decreasing and bounded by below,
and $\sigma_{t}\to0$.

Integrating the derivative of $\mathcal{U}_{t}^{*}$ and taking the
norm of both sides yields 
\begin{equation}
\|\mathcal{U}_{t}^{*}\|_{\SobSymp}\leq\|\mathcal{U}_{0}^{*}\|_{\SobSymp}+\int_{0}^{t}\|\sigma_{s}\|_{\Hsdd}\|\mathcal{U}_{s}^{*}\|_{\SobSymp}ds\,.\label{eq:estimate-U_t}
\end{equation}
The Gr\"onwall lemma, combined with $\|\mathcal{U}_{0}^{*}\|_{\SobSymp}=1$
and the estimate on $\|\sigma_{t}\|_{\Hsdd}$
provide
\[
\|\mathcal{U}_{t}^{*}\|_{\SobSymp}\leq\exp\big(\int_{0}^{t}\|\sigma_{s}\|_{\Hsdd}ds\big)\leq\exp\big(C \|\sigma_{0}\|_{\Hsdd}\big)\,.
\]
Thus, the integral $\int_{0}^{\infty}\mathcal{S}\Lambda_{s}\mathcal{U}_{s}^{*}ds$
is absolutely convergent and
\[
\mathcal{U}_{t}^{*}\xrightarrow[t\to\infty]{}\mathcal{U}_{0}^{*}-i\int_{0}^{\infty}\mathcal{S}\Lambda_{s}\mathcal{U}_{s}^{*}ds=:\mathcal{U}_{\infty}^{*}
\]
in $\SobSymp$, and
the limit $\mathcal{U}_{\infty}^{*}$ is still an implementable symplectomorphism.

Hence,
\[
\Gamma_{0}-\mathcal{U}_{\infty}\Gamma_{\infty}\mathcal{U}_{\infty}^{*}=\mathcal{U}_{t}\Gamma_{t}\mathcal{U}_{t}^{*}-\mathcal{U}_{\infty}\Gamma_{\infty}\mathcal{U}_{\infty}^{*}\to0
\]
 as $t\to\infty$, where $\Gamma_{\infty}=\big(\begin{smallmatrix}\gamma_{\infty} & 0\\
0 & 1+\bar{\gamma}_{\infty}
\end{smallmatrix}\big)$, and the convergence takes place in the space of block operators with diagonal elements in $\cHLd$ and off-diagonal elements in $\Hsdd$. This proves the last point.
\end{proof}
This completes the proof of existence.
\end{proof}

\begin{proof}[Proof of uniqueness in Prop.~\ref{thm:diagonalization-with-symplectomorphism}] 
Indeed, let us consider $\gamma'$ and $\gamma''$ satisfying the conditions
of Prop.~\ref{thm:diagonalization-with-symplectomorphism}. 
Then there exists a symplectomorphism $\mathcal{U}$ such that
$\big(\begin{smallmatrix}\gamma'' & 0\\
0 & \overline{\gamma''}+1
\end{smallmatrix}\big)=\mathcal{U}^{*}\big(\begin{smallmatrix}\gamma' & 0\\
0 & \overline{\gamma'}+1
\end{smallmatrix}\big)\mathcal{U}$. As $\cU^*\cS\cU=\cS$, this is equivalent to
\eqn \label{eqn:gamma-prime-gamma-second}
\begin{pmatrix}\gamma''+1/2 & 0\\
0 & \overline{\gamma''}+1/2
\end{pmatrix}
=\mathcal{U}^{*}\begin{pmatrix}\gamma'+1/2 & 0\\
0 & \overline{\gamma'}+1/2 \,,
\end{pmatrix}\mathcal{U}
\eeqn
and we want to prove that $\gamma'$ and $\gamma''$ are unitarily equivalent in $L^2$. 
The off-diagonal entries in \eqref{eqn:gamma-prime-gamma-second} yield $u^{*}(\gamma'+1/2)v+v^{T}(\gamma'+1/2)\bar{u}=0$ and as $\mathcal U$ is a symplectomorphism, we get from \eqref{eq:def-symplecto-alternative} that $u$ is invertible and $v\bar{u}^{-1}=u^{*-1}v^{T}$. Thus,
\[
(\gamma'+\frac{1}{2})v\bar{u}^{-1}+v\bar{u}^{-1}(\gamma'+\frac{1}{2})=0 \,.
\]
We can now use a known method to solve the Lyapunov (or Sylvester) equations:
\begin{multline*}
v\bar{u}^{-1}  = -\int_0^\infty \frac{d}{dt} \Big( e^{-t(\gamma'+\frac{1}{2})}v\bar{u}^{-1} e^{-t(\gamma'+\frac{1}{2})}\Big) dt \\
 =  \int_0^\infty e^{-t(\gamma'+\frac{1}{2})}\Big(\big(\gamma'+\frac{1}{2}\big)v\bar{u}^{-1}+v\bar{u}^{-1}\big(\gamma'+\frac{1}{2}\big)\Big) e^{-t(\gamma'+\frac{1}{2})} dt 
=0\,,
\end{multline*}
where we used that $\gamma + 1/2\geq 1/2$, so that there is no problem in handling the integrals. Hence $v=0$, and, using \eqref{eq:def-symplecto-alternative} again, $u$ is a unitary operator. And thus $\gamma''=u^{*}\gamma'u$
which proves the result.
\end{proof}

We now write the HFB equations in a form that is reminiscent of a Hamiltonian structure,
and use it to give a direct proof of the conservation of the energy.

\textbf{Notation:} For $\phi\in\Hd$, $\mathcal{U}=\big(\begin{smallmatrix}u & v\\
\bar{v} & \bar{u}
\end{smallmatrix}\big)\in \SobSymp$ a symplectomorphism, and $ \gamma_{0}^{\prime} \in \mathcal H^1$ non-negative. We set 
\begin{multline*}
\mathcal{H}_{\gamma_0^\prime}(\phi,u,v):  =\langle\phi,h\phi\rangle+\tr[(u^{*}\gamma_{0}^{\prime}u+v^{T}(1+\bar{\gamma}_{0}^{\prime})\bar{v})(h+b [|\phi\rangle\langle\phi|])]\\
 +\frac{1}{2}\tr[(u^{*}\gamma_{0}^{\prime}u+v^{T}(1+\bar{\gamma}_{0}^{\prime})\bar{v})b [u^{*}\gamma_{0}^{\prime}u+v^{T}(1+\bar{\gamma}_{0}^{\prime})\bar{v}]]\\
 +\frac{1}{2}\tr[k [u^{*}\gamma_{0}^{\prime}v+v^{T}(1+\bar{\gamma}_{0}^{\prime})\bar{u}+|\phi\rangle\langle\bar{\phi}|]
 (v^{*}\gamma_{0}^{\prime}u+u^{T}(1+\bar{\gamma}_{0}^{\prime})\bar{v}+|\bar{\phi}\rangle\langle\phi|)] \,.
\end{multline*}
In the next proposition and its proof we use the  abbreviations $h (t)\equiv h (\gamma_t^{\phi_{t}})$ and $k (t)\equiv  k (\sigma^{\phi_{t}}_t),$ where, recall, $\gamma^{\phi}:=\gamma+|\phi\rangle\langle\phi|$ and $\sigma^{\phi}:=\sigma +\phi\otimes\phi$, and 
$ h (\gamma)$ and $  k (\sigma)$ are defined in \eqref{h} and \eqref{k}.

\begin{proposition}\label{prop:quasi-hamiltonian-structure}
Let 
$\rho_t=(\phi_t,\gamma_t,\sigma_t)\in C^{0}([0,T);X^{3})\cap C^{1}([0,T);\XL)
$ be a solution to the HFB equations \eqref{eq:BHF-phi}$\sim$\eqref{eq:BHF-sigma} in the classical sense, on an interval $[0,T)$, with $T>0$. 
Let $\mathcal{U}_{t}$ and $\gamma_0^\prime$ be as in Proposition~\ref{prop:symplectic-structure}.

Then $\mathcal{E}(\phi_{t},\gamma_{t},\sigma_{t})=\mathcal{H}_{\gamma_0^\prime}(\phi_{t},u_{t},v_{t})$ and the derivatives
of $\mathcal{H}_{\gamma_0^\prime}$ and of $(\phi_{t},u_{t,}v_{t})$
are linked through the equations
\begin{align}
\frac{\partial\mathcal{H}_{\gamma_0^\prime}}{\partial\langle\phi|}(\phi_{t},u_{t},v_{t}) & =\phantom{-\gamma_{0}^{\prime} \,}i\partial_{t}\phi_{t}\label{eq:dH-dPhi-star} \,, \allowdisplaybreaks \\
\frac{\partial\mathcal{H}_{\gamma_0^\prime}}{\partial u^{*}}(\phi_{t},u_{t},v_{t}) & =\phantom{-}\gamma_{0}^{\prime} \, i\partial_{t}u_{t}+\frac{1}{2}v_{t}\overline{k (t)}\label{eq:dH-du-star} \,, \allowdisplaybreaks \\
\frac{\partial\mathcal{H}_{\gamma_0^\prime}}{\partial v^{*}}(\phi_{t}, u_{t},v_{t}) & =-\gamma_{0}^{\prime} \, i\partial_{t}v_{t}+v_{t}\overline{h (t)}+\frac{1}{2}u_{t}k (t)\label{eq:dH-dv-star} \,.
\end{align}
The conservation of the energy $\mathcal{E}(\phi_{t},\gamma_{t},\sigma_{t})$
follows.
\end{proposition}
\begin{proof}Eq.~\eqref{eq:symplectic-transfo-Gamma} is equivalent to
\begin{align*}
\gamma_{t} & =u_{t}^{*}\gamma_{0}^{\prime}u_{t}+v_{t}^{T}(1+\bar{\gamma}_{0}^{\prime})\bar{v}_{t} \,, \\
\sigma_{t} & =u_{t}^{*}\gamma_{0}^{\prime}v_{t}+v_{t}^{T}(1+\bar{\gamma}_{0}^{\prime})\bar{u}_{t} \,.
\end{align*}
Hence, we can rewrite the expression of the energy in terms of $\phi_t$,
$u_t$, and $v_t$ as $\mathcal{E}(\phi_{t},\gamma_{t},\sigma_{t})=\mathcal{H}_{\gamma_0^\prime}(\phi_{t},u_{t},v_{t})$.
We then compute the derivatives of $\mathcal H_{\gamma_0^\prime}$:
\begin{align*}
\frac{\partial\mathcal{H}_{\gamma_0^\prime}}{\partial\langle\phi|}(\phi,u,v) & =h\phi+b [u^{*}\gamma_{0}^{\prime}u+v^{T}(1+\bar{\gamma}_{0}^{\prime})\bar{v}]\phi+k [\sigma+\phi\otimes\phi]|\bar{\phi}\rangle  \,, \allowdisplaybreaks\\
\frac{\partial\mathcal{H}_{\gamma_0^\prime}}{\partial u^{*}}(\phi,u,v) & =\gamma_{0}^{\prime}u(h+b [|\phi\rangle\langle\phi|]+b [u^{*}\gamma_{0}^{\prime}u+v^{T}(1+\bar{\gamma}_{0}^{\prime})\bar{v}]) \\
 & \quad+(\frac{1}{2}+\gamma_{0}^{\prime})v k [v^{*}\gamma_{0}^{\prime}u+u^{T}(1+\bar{\gamma}_{0}^{\prime})\bar{v}+|\bar{\phi}\rangle\langle\phi|] \,, \allowdisplaybreaks\\
\frac{\partial\mathcal{H}_{\gamma_0^\prime}}{\partial v^{*}}(\phi,u,v) & =(1+\gamma_{0}^{\prime})v(\bar{h}+b [|\bar{\phi}\rangle\langle\bar{\phi}|]+b [u^{T}\bar{\gamma}_{0}^{\prime}\bar{u}+v^{*}(1+\gamma_{0}^{\prime})v])\\
 & \quad+(\frac{1}{2}+\gamma_{0}^{\prime})u k [u^{*}\gamma_{0}^{\prime}v+v^{T}(1+\bar{\gamma}_{0}^{\prime})\bar{u}+|\phi\rangle\langle\bar{\phi}|] \,.
\end{align*}
Replacing $(\phi,u,v)$ by $(\phi_{t},u_{t},v_{t})$ yields
\begin{align*}
\frac{\partial\mathcal{H}_{\gamma_0^\prime}}{\partial\langle\phi|}(\phi_{t},u_{t},v_{t}) & =h\phi_{t}+b [\gamma_{t}]\phi_t+k (t)\bar{\phi}_t \,, \allowdisplaybreaks\\
\frac{\partial\mathcal{H}_{\gamma_0^\prime}}{\partial u^{*}}(\phi_{t},u_{t},v_{t}) & =\gamma_{0}^{\prime}u_{t}h (t)+(\frac{1}{2}+\gamma_{0}^{\prime})v_{t}\overline{k (t)} \,, \allowdisplaybreaks\\
\frac{\partial\mathcal{H}_{\gamma_0^\prime}}{\partial v^{*}}(\phi_{t},u_{t},v_{t}) & =(1+\gamma_{0}^{\prime})v_{t}\overline{h (t)}+(\frac{1}{2}+\gamma_{0}^{\prime})u_{t}k (t) \,,
\end{align*}
which are in fact \eqref{eq:dH-dPhi-star}, \eqref{eq:dH-du-star}, \eqref{eq:dH-dv-star} using the HFB equations. 
Hence, using first the chain rule, then (\ref{eq:dH-dPhi-star}),
(\ref{eq:dH-du-star}), and (\ref{eq:dH-dv-star}),
\begin{align*}
\frac{d}{dt}\mathcal{H}_{\gamma_0^\prime}(\phi_{t},u_{t},v_{t}) & =\langle\partial_{t}\phi_{t}|\frac{\partial\mathcal{H}_{\gamma_0^\prime}}{\partial\langle\phi|}(\phi_{t},u_{t},v_{t})+\frac{\partial\mathcal{H}_{\gamma_0^\prime}}{\partial|\phi\rangle}(\phi_{t},u_{t},v_{t})|\partial_{t}\phi_{t}\rangle\\
 & \quad+\tr[\partial_{t}u_{t}^{*}\frac{\partial\mathcal{H}_{\gamma_0^\prime}}{\partial u^{*}}(\phi_{t},u_{t},v_{t})]+\tr[\partial_{t}u_{t}\frac{\partial\mathcal{H}_{\gamma_0^\prime}}{\partial u}(\phi_{t},u_{t},v_{t})]\\
 & \quad+\tr[\partial_{t}v_{t}^{*}\frac{\partial\mathcal{H}_{\gamma_0^\prime}}{\partial v^{*}}(\phi_{t},u_{t},v_{t})]+\tr[\partial_{t}v_{t}\frac{\partial\mathcal{H}_{\gamma_0^\prime}}{\partial v}(\phi_{t},u_{t},v_{t})]\\
 & =\re \tr[\partial_{t}u_{t}^{*}v_{t}\overline{k (t)}+\partial_{t}v_{t}^{*}(v_{t}\overline{h (t)}+\frac{1}{2}u_{t}k (t))] \,.
\end{align*}
We can now use that the evolution equation \eqref{eq:evolution-symplecto} on $\mathcal U_t$   is equivalent to
\begin{align}
i\partial_{t}u_{t} & =u_{t}h (t)+v_{t}\overline{k (t)}\label{eq:HFB-u} \,, \\
i\partial_{t}v_{t} & =-u_{t}k (t)-v_{t}\overline{h (t)}\label{eq:HFB-v} \,,
\end{align}
along with $\tr[A^{T}]=\tr[A]$
and the cyclicity of trace to group all the terms as in
\begin{multline*}
\frac{d}{dt} \mathcal{H}_{\gamma_0^\prime}(\phi_{t},u_{t},v_{t}) 
=\im \tr[\overline{k (t)}h (t)(v_{t}^{T}\bar{u}_{t}-u_{t}^{*}v_{t})-\overline{k (t)}k (t)v_{t}^{*}v_{t}\\
+\overline{k (t)}h (t)v_{t}^{T}\bar{u}_{t}+2h (t)h (t)v_{t}^{T}\bar{v}_{t}+\overline{k (t)}k (t)u_{t}^{T}\bar{u}_{t}+h (t)k (t)u_{t}^{T}\bar{v}_{t}]
\end{multline*}
which then vanishes since $v_{t}^{T}\bar{u}_{t}=u_{t}^{*}v_{t}$ for
a symplectomorphism (see \eqref{eq:def-symplecto-alternative}), and the terms $\overline{k (t)}k (t)v_{t}^{*}v_{t}$,
$h (t)h (t)v_{t}^{T}\bar{v}_{t}$, $\overline{k (t)}k (t)u_{t}^{T}\bar{u}_{t}$,
and $\overline{k (t)}h (t)v_{t}^{T}\bar{u}_{t}+h (t)k (t)u_{t}^{T}\bar{v}_{t}$
give real traces. 
\end{proof}

\section{Relation with the HFB eigenvalue equations}
\label{sec:Relation_with_HFB_eigengalues_equations}

In this section, we link our work with the HFB eigenvalue equations
often encountered in the physics literature \cite{Griffin1996, DoddEdwardsClarkBurnett1998, ParkinsWalls1981}.

To be explicit, we give, in Table~\ref{tab:Correspondance-Griffin-BBCFS},
the correspondence between the notations of this article and those
of an article of Griffin \cite{Griffin1996}. 
\begin{table}
\begin{centering}
\begin{tabular}{|c||c|c|c|c|c|c|c|}
\hline 
this article & $\phi(x)$ & $\gamma(x;x)$ & $\sigma(x,x)$ & $h_{g\delta}$ & $k_{g\delta}$ & $N_{j}$ & $V$\tabularnewline
\cite{Griffin1996} & $\Phi(\boldsymbol{r})$ & $\tilde{n}(\boldsymbol{r})$ & $\tilde{m}(\boldsymbol{r})$ & $\hat{\mathscr{L}}$ & $gm(\boldsymbol{r})$ & $N_{0}(E_{j})$ & $U_{ext}-\mu$\tabularnewline
\hline 
\end{tabular}
\par\end{centering}

\caption{\label{tab:Correspondance-Griffin-BBCFS}Correspondence between the
notations of this article and some notations common in the physics
literature \cite{Griffin1996}. }
\end{table}
We note that the setting in \cite{Griffin1996} is not exactly the same as ours, since the class of external
potentials $V$ that we consider excludes trapping potentials, and the
solutions  $\Phi(\boldsymbol{r})$ considered in  \cite{Griffin1996} are
time-independent.
Moreover, we note that in this paper, we give rigorous proofs in the case
of a two-body interaction potential $v$ such that $v^{2}$ is relatively
form-bounded with respect to the Laplacian, which excludes potentials
as singular as $g\delta$; hence, the correspondence we establish in
this section is only formal. 
Nevertheless, we believe that pointing out this relationship is useful.

Moreover, we note that in the physics literature (see
e.g.,~\cite[(23)]{Griffin1996}), the HFB eigenvalue equations are often investigated using a
generalized eigenbasis decomposition (using vectors often denoted by $u_j$, $v_j$ which play the same
role as below), which we
can relate to our approach in the following manner,
based on our discussion from Section~\ref{sec:symplectic-hamiltonian-structure}.

Let $\mathcal{U}_{t}=\big(\begin{smallmatrix}u_{t} & v_{t}\\
\bar{u}_{t} & \bar{v}_{t}
\end{smallmatrix}\big)$, and let $\gamma_{0}^{\prime}\geq0$ be a trace class operator as in Prop.~\ref{prop:symplectic-structure},
with the orthonormal decomposition $\gamma_{0}^{\prime}=\sum_{j\geq0}N_{j}|\zeta_{j}\rangle\langle\zeta_{j}|$.
Let 
\[
u_{j,t}:=u_{t}^{*}\zeta_{j}\qquad\mbox{and}\qquad v_{j,t}:=-v_{t}^{*}\zeta_{j}\,.
\]
Then \eqref{eq:symplectic-transfo-Gamma} yields
\begin{align*}
\gamma_{t} & =\sum_{j\geq0}\big(N_{j}\,|u_{j,t}\rangle\langle u_{j,t}|+(1+N_{j})\,|\bar{v}_{j,t}\rangle\langle\bar{v}_{j,t}|\big)\,,\\
\sigma_{t} & =\sum_{j\geq0}\big(N_{j}\,|u_{j,t}\rangle\langle v_{j,t}|+(1+N_{j})\,|\bar{v}_{j,t}\rangle\langle\bar{u}_{j,t}|\big)\,.
\end{align*}
which yield \cite[(25)]{Griffin1996} by evaluation on the diagonal:
\begin{align}
\gamma_{t}(x;x) & =\sum_{j\geq0}\big(N_{j}\,|u_{j,t}(x)|^{2}+(1+N_{j})\,|v_{j,t}(x)|^{2}\big)\,,\label{eq:gamma-of-u_j-v_j}\\
\sigma_{t}(x,x) & =\sum_{j\geq0}u_{j,t}(x)\bar{v}_{j,t}(x)(1+2N_{j})\,.\label{eq:sigma-of-u_j-v_j}
\end{align}

We now consider a pair interaction potential $v=g\delta$. We assume that
$\phi$ is independent of time and $u_{j,t}$, $v_{j,t}$ have the
simple form 
\begin{equation}
u_{j,t}=e^{-itE_{j}}u_{j,0}\,,\qquad v_{j,t}=e^{-itE_{j}}v_{j,0}\,.\label{eq:v_j-of-t}
\end{equation}
We also distinguish the quantities corresponding to $v=g\delta$ by the index $g\delta$. Then \eqref{eq:evolution-symplecto} formally yields the HFB eigenvalue equations 
\begin{align*}
h_{g\delta}u_{j}-k_{g\delta}v_{j} & =E_{j}u_{j}\,,\\
\overline{h_{g\delta}}v_{j}-\overline{k_{g\delta}}u_{j} & =-E_{j}v_{j}\,,
\end{align*}
as presented in the work of Griffin \cite[Eq.~(23)]{Griffin1996}.
Note that (\ref{eq:gamma-of-u_j-v_j}), (\ref{eq:sigma-of-u_j-v_j}),
and (\ref{eq:v_j-of-t}) imply that $\gamma_{t}(x;x)$ and $\sigma_{t}(x;x)$
are time independent, since the phases simplify.

We conclude that the HFB eigenvalue equations are the stationary version of our 
equation~\eqref{eq:evolution-symplecto}. It amounts to finding eigenvalues and eigenvectors for the
matrix $\Lambda\mathcal{S}$ in \eqref{eq:evolution-symplecto}, which is a
nonlinear problem since $\Lambda$ depends on $\gamma$ and $\sigma$ (that is, on $u$, $v$ and $\gamma_0^\prime$). 
Furthermore, the decomposition in functions $u_{j}$ and $v_{j}$ corresponds to a ``diagonalization'' 
of the generalized one-particle density matrix $\Gamma$ in the sense of 
Proposition~\ref{thm:diagonalization-with-symplectomorphism}.

\section{Existence and Uniqueness of Solutions to the HFB Equations}
\label{sec-GWP-HFB-1}

We prove the global in time existence and uniqueness of mild 
solutions to the time-dependent Hartree-Fock-Bogoliubov equations in the $H^{1}$-setting. 

We recall that, given a Banach space $X$, $f\in C(X)$, a continuous function on $X$, and $-iA$ the infinitesimal generator of a strongly continuous semigroup $G(t)$ on $X$, 
 a continuous function $\rho:[0,T)\to X$ is called a \textit{mild solution} of the problem
 \begin{align}\label{eq:def_problem}
\begin{cases}
i\partial_t \rho & =A\rho+f(\rho)\,,\\
\rho(0) & =\rho_{0}\in X\,,
\end{cases}
\end{align}
if $\rho_t$ solves the fixed point equation in integral form  (with the integral in Bochner's sense)\begin{equation}\label{eq:def_mild_HFB}
\rho_t=G(t)\rho_0 -i\int_0^t G(t-s) f(\rho_s)\,ds.
\end{equation}

In what follows we use the notation  $A\lesssim B$ to stand for an inequality of the form $A\leq CB$, for some constant where $C>0$. 
The main result of this section is the following 

\begin{theorem}
\label{thm:existence-uniquenes-up-to-Coulomb}
Let $d\leq3$ and $\rho_{0}=(\phi_{0},\gamma_{0},\sigma_{0})\in \XL$. 
Assume that the potentials $V$ and $v$ satisfy Conditions (a) and (b') of Subsection \ref{sec:quant-mb-probl}. 
\DETAILS{\begin{itemize}
\item $V$ is infinitesimally 
bounded with respect to the Laplacian,
\item $v\in W^{p, 1}(\R^d)$ for some $p >  d$. 
\end{itemize}}
Then the following hold:
\begin{enumerate}
\item[(i)] \label{ite:local-existence-mild}
\textit{Existence and uniqueness of a local mild solution:}

There exists a unique maximal solution
\[
(\rho_{t})_{t\in[0,T)}=(\phi_{t},\gamma_{t},\sigma_{t})_{t\in[0,T)}\in C^{0}([0,T);\XL)
\]
to the HBF equations \eqref{eq:BHF-phi} to \eqref{eq:BHF-sigma} 
in the mild sense, for some $0<T\leq \infty$.

\item[(ii)]\label{ite:local-existence-classical}
\textit{Existence and uniqueness of a local classical solution:}

If $\rho_{0}\in X^{3}$, then 
\[
(\rho_{t})_{t\in[0,T)}\in C^{0}([0,T);X^{3})\cap C^{1}([0,T);\XL)
\]
and $\rho_{t}$ satisfies the HBF equations \eqref{eq:BHF-phi} to \eqref{eq:BHF-sigma} in the classical sense.

\item[(iii)]\label{ite:conservation-laws}
\textit{Conservation laws:}

The number of particles $\tr[\gamma_t]$ and the energy \eqref{eq:energy} are constants.

\item[(iv)] \textit{Positivity preservation property}: \\ 
If $\Gamma=\begin{pmatrix} \gamma & \sigma\\
 \bar{\sigma} & 1+\bar \gamma \end{pmatrix}\geq 0$ at $t=0$, 
then this holds for all times.

\item[(v)]\label{ite:global-existence}
\textit{Existence of a global solution:} \\
If additionally $\Gamma_0\geq 0$, then the solution 
$\rho_{t}$ is global, i.e., $T=\infty$.
\end{enumerate}
\end{theorem}


%

\begin{proof}[Proof of Theorem~\ref{thm:existence-uniquenes-up-to-Coulomb}(i) 
 {[Local Mild Solutions]}]
We use the notations introduced at the beginning of Section~\ref{sec-HFB-general-1}. 
The proof is based on a standard fixed point argument (through an application of
the Cauchy-Lipschitz and Picard-Lindel\"of theorem).  
Separating the linear part $A\rho$ and nonlinear part $f(\rho)$, we can write the HFB equations \eqref{eq:BHF-phi} to \eqref{eq:BHF-sigma} in the form
\eqn
i\partial_{t}\rho=A\rho +f(\rho)\,,\label{eq:HBFv_form_in_rho}
\eeqn
where $\rho:=(\phi,\gamma,\sigma)\in X^{2}$. 
Then the linear part in the HFB equations is given by  
\eqn \label{eq:def-A}
A\rho=\big( h\phi \,,\, [h,\gamma] \,,\, [h,\sigma]_{+}+k[\sigma] \big)\,,
\eeqn
with the domain $D(A)=X^2$, and the nonlinear part $f:=(f_{1},f_{2},f_{3})$ by
\begin{align}
\label{eq:def_f_1}
f_{1}(\rho) & =b[\gamma]\phi+k[\sigma+\phi^{\otimes2}]\bar{\phi}\,, \\
\label{eq:def_f_2}
f_{2}(\rho) & =[b[\gamma+|\phi\rangle\langle\phi|],\gamma]+k[\sigma+\phi^{\otimes2}]\bar{\sigma}-\sigma\overline{k[\sigma+\phi^{\otimes2}]}\,, \\
\label{eq:def_f_3}
f_{3}(\rho) & =[b[\gamma+|\phi\rangle\langle\phi|],\sigma]_{+}+[k[\sigma+\phi^{\otimes2}],\gamma]_{+}\,.
\end{align}

From Lemma~\ref{lem:f_C1}, below, we obtain that $f$ 
 is continuously Fr\'echet differentiable in $\XL$ and therefore is locally Lipschitz, 
and from Lemma~\ref{lem:G(t)}, we obtain that $G(t)=\exp(itA)$ defines a strongly continuous 
 uniformly bounded semigroup on~$\XL$. 

Consequently, we can rewrite the HFB equations \eqref{eq:BHF-phi} - \eqref{eq:BHF-sigma} as a fixed point problem
\[
\rho_t=G(t) \rho_0 -i\int_0^t G(t-s) f\big(\rho_s\big) \, ds \,.
\] 
and use the Banach contraction theorem to show that \eqref{eq:BHF-phi} - \eqref{eq:BHF-sigma} have the unique local mild solution to in $\XL$
for the given initial data. (For the details for this standard argument, see \cite[Sect.~9.2e, Thm~3]{McOwen2003}.)
\end{proof}

We will now prove our main Lemmata on $G(t)=\exp(itA)$ and $f$. First, we recall the norms \eqref{norms}. 
Moreover,  if we denote the integral kernel of an operator $\s$ by $\tilde\s$, then the  norm $\|\sigma\|_{\mathcal{H}^j_\s}$ is equivalent to the  norm
\[ \|\sigma\|_{\mathcal{H}^j_\s}\simeq \|\tilde\s\|_{H^j}:= \|(M^{2}\otimes1+1\otimes M^{2})^{j/2}\tilde\s)\|_{L^{2}(\mathbb{R}^{2d})}\,.
\]
\begin{lemma}\label{lem:G(t)}
The operator $A$ generates a strongly continuous semigroup, $G(t)=\exp(itA)$, on $\XL$, uniformly bounded as $\|G(t)\|_{\cB(\XL)}\le 1$. 
\end{lemma}
%
\begin{proof} 
Let $\hat h(\sigma):= [h,\sigma]_{+}+k[\sigma]$. We define $G(t)=\exp(itA)$  on $\rho:=(\phi,\gamma,\sigma)\in X^{j}$ as  \eqn \label{G(t)-def}
G(t)\rho:=(\exp(-ith)\phi,\exp(-ith)\gamma\exp(ith), \exp(-it\hat h)(\sigma)). \eeqn 
We use that $-\Delta$ is $h$-bounded, and $h$ is $-\Delta$-bounded and that $M$ is translationally invariant. For $
(\phi,\gamma,\sigma)\in \XL$ and $k$ s.t. $M\le h+k$,
\begin{equation*}
\|\exp(-ith)\phi\|_{\Hd} = \|(h+k)\exp(-ith)\phi\|_{L^2}  = \|(h+k)\phi\|_{L^2}^{2} \ls \|\phi\|_{\Hd}. 
\end{equation*}
Similarly
$
\|\exp(-ith)\gamma\exp(ith)\|_{\cHLd}\ls \|\gamma\|_{\cHLd}\,.
$

Finally, we define the operator $\tilde h$ acting on $L^2(\R^{2d})$ by the condition $\widetilde{\hat h(\sigma)}= \tilde{h}\tilde\sigma$. 
Then we have $\tilde{h}=h_x+h_y+v(x-y)$, since  the pair potential $v$ be infinitesimally bounded with respect to $-\Delta$, the operator  $\tilde{h}=h_x+h_y+v(x-y)$ 
is self-adjoint and $h$ and $-\Delta_x-\Delta_y$ are mutually relatively bounded. Hence, using \eqref{G(t)-def} and choosing $c$ s.t. $M_x+M_y\le \tilde{h}+c$,
\begin{align}\notag \|\exp(-it([h,\sigma]_{+}+k[\sigma]))\sigma\|_{\mathcal{H}^1_\s}&\simeq \|\exp(-it\tilde{h})\tilde\s\|_{H^1}\\
\notag &\ls \|(\tilde{h}+c)\exp(-it\tilde{h})\tilde\s\|_{L^2}\ls \|\tilde\s\|_{H^1}\simeq \|\sigma\|_{\mathcal{H}^1_\s}\,.
\end{align} 
The strong continuity of $G(t)$ follows from the strong continuity of $\exp(-ith)$ and $\exp(-it\hat h)$. 
\end{proof}  

The 
following lemma allows us to control the nonlinear term $f$ in the HFB equations.

\begin{lemma}\label{lem:f_C1} 
The vector of nonlinear terms $f=(f_1,f_2,f_3)$ defined in Eq.~\eqref{eq:def_f_1}$-$\eqref{eq:def_f_3}   
maps $\XL$ into itself and is continuously Fr\'echet differentiable in $\XL$ ($f\in C^{1}(\XL)$).
\end{lemma}
\begin{proof}[Proof of Lemma~\ref{lem:f_C1}] 
For the first statement it is sufficient to prove that, for the quadratic and cubic  parts of $f$ are bounded as  
\begin{align}\label{f-est-quadr}
&\big\|\big( b[\gamma]\phi+k[\sigma]\bar{\phi}\,,\,[b[\gamma],\gamma]+k[\sigma]\bar{\sigma}-\sigma\overline{k[\sigma]}\,,
[b[\gamma],\sigma]_{+}+[k[\sigma],\gamma]_{+}\big)\big\|_{\XL}\ls \|\rho\|_{\XL}^{2}\,,\\
&\notag \big\|\big( k[\phi^{\otimes2}]\bar{\phi}\,,\,[b[|\phi\rangle\langle\phi|],\gamma]+k[\phi^{\otimes2}]\bar{\sigma}-\sigma\overline{k[\phi^{\otimes2}]}\,,\\
\label{f-est-cub} &\qquad \qquad \qquad \qquad \qquad \qquad \qquad \qquad   [b[|\phi\rangle\langle\phi|],\sigma]_{+}+[k[\phi^{\otimes2}],\gamma]_{+}\big)\big\|_{\XL}\ls \|\rho\|_{\XL}^{3}\,.
\end{align}
All the cubic estimates can be deduced
from their quadratic counterparts using
\[
\||\phi\rangle\langle\phi|\|_{\cH^1_\g}\leq\|\phi\|_{\Hd}^{2}
\quad \text{and} \quad
\|\phi\otimes\phi\|_{\Hd}\leq\|\phi\|_{\Hd}^{2}\,.
\]
We thus only consider the quadratic terms.
Using Lemma~\ref{prop:continuity_of_B_and_K}\eqref{enu:continuity-B-H1L2-LH1}, we estimate
\begin{align*}
\|b[\gamma]\phi\|_{\Hd} 
\ls \|\gamma\|_{\cHLd}\|\phi\|_{\Hd}\ls \|\rho\|_{\XL}^{2}\,.
\end{align*}
%
 For~$k[\sigma]\bar{\phi}$, we use Lemma~\ref{prop:continuity_of_B_and_K}\eqref{enu:continuity-K-H1L2-L2H1} to find
\[
\|k[\sigma]\bar{\phi}\|_{\Hd}
\leq\|M k[\sigma]\|_{\cB}\|\bar{\phi}\|_{L^2}\ls \|\sigma\|_{\Hsdd}\|\phi\|_{L^2}\,.
\]

 We estimate $[b[\gamma],\gamma]$ using Lemma~\ref{prop:continuity_of_B_and_K}.(\ref{enu:continuity-B-H1L2-LH1})
\begin{align*}
\|[b[\gamma],\gamma]\|_{\cHLd} \leq 2 &\|M b[\gamma]M^{-1}M\gamma M\|_{\mathcal L^1} \\
& \leq  2\|b[\gamma]\|_{\cHLd} \|\gamma \|_{ \cHLd}\ls  \|\gamma\|^2_{\cHLd}
 \ls \|\rho \|_{ \XL}^2\,.
\end{align*}

For $k[\sigma]\bar{\sigma}$ (and similarly $\sigma\overline{k[\sigma]}$), 
the inequality $$\|k[\sigma]\bar\s\|_{\cH^{1}_\g}=\|Mk[\sigma]\bar\s M\|_{\cL^{1}}\le\|M k[\sigma]\|_{\cL^{2}}\|\bar\s M\|_{\cL^{2}},$$ Lemma~\ref{prop:continuity_of_B_and_K}\eqref{enu:continuity-K-H1L2-L2H1} (see estimate \eqref{Mk-est} below)  
and $\|\bar\s M\|_{\cL^{2}}\le \|\s\|_{\cH_{\s}^{1}}$ (which follows from the definition of $\|\s\|_{\cH_{\s}^{1}}$) give the estimate 
\begin{align}\label{ksig-est}
\|k[\sigma]\bar{\sigma}\|_{\cHLd} 
\ls \|\sigma\|_{\Hsdd}^{2} \,.  \end{align}

For $b[\gamma]\sigma$ (or similarly $\sigma\overline{b[\gamma]}$), using Lemma~\ref{prop:continuity_of_B_and_K}\eqref{enu:continuity-B-H1L2-LH1}, we obtain
\begin{align*}
\|b[\gamma]\sigma\|_{\Hsdd}  \leq\|M b[\gamma]M^{-1}\|_{\cB}\|M\sigma M\|_{\mathcal{L}^{2}}
  \ls \|\gamma\|_{\cHLd}\|\sigma\|_{\Hsdd} \,.
\end{align*}
And finally $k[\sigma]\bar{\gamma}$ (and similarly $\gamma k[\sigma]$), using Lemma~\ref{prop:continuity_of_B_and_K}.\eqref{enu:continuity-K-H1L2-L2H1} (see estimate \eqref{Mk-est}), we arrive at
\begin{align*}
\|k[\sigma]\bar{\gamma}\|_{\Hsdd}  \leq\|M k[\sigma]\|_{\mathcal{L}^{2}}\|\bar{\gamma}M\|_{\cB}
  \ls \|\sigma\|_{\Hsdd}\|\gamma\|_{\cHLd}\,,
\end{align*}
which completes the proof of \eqref{f-est-quadr} and therefore of \eqref{f-est-cub}.  

To prove that $f$ is Fr\'echet differentiable, we observe that each $f_j$ is a linear combination of multi-linear maps and therefore $df(\rho)\xi$ is of the same form as $f(\rho)$ and can be estimated as above.  \end{proof}

\begin{proof}[Proof of Theorem~\ref{thm:existence-uniquenes-up-to-Coulomb}(ii) 
 {[Local Classical Solutions]}]
The existence of classical solutions to the HFB equations for initial data in $X^3$ then follows from:
\begin{lemma}[See {\cite[Lemma 3.1]{Segal1963}}.]
\label{lem:from_mild_to_classical}
If $-iA$ is the generator of a continuous one-parameter semi-group in the Banach space $X$, and if $f$ is continuously differentiable on $X$, then a mild solution of Eq.~\eqref{eq:def_problem} has its values in the domain $\mathcal D(A)$ of $A$ throughout its interval of existence provided this is the case initially.

In other words, $\rho_t$, if it exists at all, then satisfies the differential equation \eqref{eq:def_problem}  in the obvious sense.\qedhere
\end{lemma}
\end{proof}

\begin{proof}[Proof of Theorem~\ref{thm:existence-uniquenes-up-to-Coulomb}(iii) 
 {[Conservation Laws]}]
For classical solutions, the conservation of the number particle and of the energy were proven as consequences of the same conservation laws for the many body system in Theorem \ref{prop-en-vfull-1} and \ref{thm:preservation-particle-number}. Another proof of the conservation law for the energy  using only the HFB equations (independently from the many body problem) was given in Prop.~\ref{prop:quasi-hamiltonian-structure}, and the conservation of the particle number could also be proven directly from \eqref{eq:BHF-gamma}. We can now use those results since we proved the local existence of a classical solution. The conservation laws then extend to mild solutions by approximation.
\end{proof}

\begin{proof}[Proof of Theorem~\ref{thm:existence-uniquenes-up-to-Coulomb}(iv) {[Positivity preservation property]}] This follows from relation \eqref{Gam-evol}. Indirectly, it follows from the equivalence of the HFB equations \eqref{eq:BHF-phi} to~\eqref{eq:BHF-sigma} the self-consistent equation \eqref{eq-omt-commut-3} (see Theorem \ref{thm-Hquadr-dyn-1}). \end{proof}

\begin{proof}[Proof of Theorem~\ref{thm:existence-uniquenes-up-to-Coulomb}(v) 
 {[Global Solution]}]
We recall that for a maximal solution $\rho_t$ of the mild problem \eqref{eq:def_mild_HFB} defined on an interval $[0,T)$, we have that either $T=\infty$ or $\sup_{t\in [0,T)}\|\rho_t\|_{\XL}=\infty$ (see, e.g., \cite[Thm~4.3.4]{CazenaveHaraux1998}). It is thus enough to prove that
\[
\sup_{t\in [0,T)} \big\{
\|\phi_t\|_{\Hd},
\|\gamma_t\|_{\cHLd},
\|\sigma_t\|_{\Hsdd}\big\}<\infty
\] to show that the solutions are global. 
Let
\eqn \label{eqn:def-mathbb-T}
\mathbb T := \int dxdy \; \psi^*(x)(-\Delta)\psi(y) \,,
\eeqn
Because $V$ is infinitesimally form bounded with respect to the Laplacian,
\eqn\label{eqn:estimate-dGamma-V}
 \int dx \; \psi^*(x)\psi(x)V(x)\geq -\frac{1}{2}\mathbb T - c\opN
\eeqn
holds. 
And, because the pair potential $v$ is bounded, we have
\eqn \label{eqn:estimate-dGamma2-v}
\mathbb V := \frac{1}{2}\int dx dy \; v(x-y)\psi^*(x) \psi^*(y)
    \psi(x) \psi(y) \geq  
    -C\opN^{2}    -C\opN\,.
\eeqn
Hence, from the definition of $\opH$, \eqref{eqn:estimate-dGamma-V} and \eqref{eqn:estimate-dGamma2-v} we get
\begin{align}
\mathbb T \leq 2\opH + C \opN^{2}  + C \opN\,.
\end{align}
We now take the expectation value of $\qf_t$ and use that $\qf_t$ is quasifree to bound $\qf_t(\opN^{2})$ by $C(\qf_t(\opN)^{2}+1)$ and the conservation of the particle number and of the energy to obtain  
\begin{align}\notag
\tr[-\Delta(\gamma_t+|\phi_t\rangle \langle\phi_t|)] & \leq C(\mathcal E(\phi_t,\gamma_t,\sigma_t)+\sum_{k=0}^2\cN(\phi_t,\gamma_t,\sigma_t)^{k}) \\
\label{gam-E-bnd} & \leq C(\mathcal E(\phi_0,\gamma_0,\sigma_0)+\sum_{k=0}^2\cN(\phi_0,\gamma_0,\sigma_0)^{k}) \,.
\end{align}
Combined with the conservation of the particle number, this estimate provides 
bounds on $\|\gamma_t\|_{\cHLd}$ and $\|\phi_t\|_{\Hd}$ that are uniform in $t$. 
Moreover, uniform bounds on $\|\sigma_t\|_{\Hsdd}$ are then obtained from Proposition~\ref{prp-gampos-sigsym-1}. 
It thus follows that the solution is global, as claimed.
\end{proof}


\section{Gibbs states and Bose-Einstein condensation}
\label{sec-Gibbs-HFB-1}

In this section, we determine translation- and $U(1)$ gauge-invariant  Gibbs states
for the HFB equations without an external potential, and with an interaction potential $g\delta$, and discuss the emergence of a Bose-Einstein condensate at positive temperature.  (Recall from the introduction that  $U(1)$ gauge-invariant  Gibbs states for the HFB equations are, in fact, Gibbs states for the Hartree-Fock equations.)

We consider the system
on a torus, $\Lambda_L=\R^d/2L \Z^d$, i.e., $[- L, L]^d$ with periodic boundary conditions.
Accordingly, we denote $\Lambda_L^*:=\frac{\pi}{L}\Z^d$
the lattice reciprocal to $2L \Z^d$. We will eventually take the thermodynamic limit, $L\rightarrow\infty$,
and discuss the emergence of a Bose-Einstein condensate.

The Hamiltonian  $\opH$ of the Bose gas is $U(1)$ gauge-invariant (that is, invariant under the transformation $\psi^\sharp\rightarrow (e^{i\theta}\psi)^\sharp$), and, as we consider the case with no external potential, translation invariant. On a compact torus, where the volume is finite, these symmetries are also present in the Gibbs states of system (the notion of translation invariance should be, of course, appropriately modified). We are interested in quasifree states $\qf_L$ which on the one hand satisfy both the $U(1)$ gauge invariance and the translation invariance, and, on the other hand satisfy a fixed point equation  corresponding to the consistency condition~\eqref{eq-omt-commut-3} in the dynamical case:
\eqn\label{eqn:fixed-point-Gibbs-quasifree-state}
\Phi(\qf_L)=\qf_L \quad\text{with}\quad
\Phi(\qf_L)(\opA):=\tr[\opA\,\exp(-\beta(\opH_{HFB}(\qf_L)-\mu\opN))/\Xi]
\eeqn
where $\beta>0$ is the inverse temperature, 
$\mu$ is the chemical potential, and $\Xi=\tr[\exp(-\beta(\opH_{HFB}(\qf_L))-\mu\opN)]$.  
The  $U(1)$ gauge-invariance of $\qf_L$ then implies that the truncated expectations $\phi_{\qf_L}$ and $\sigma_{\qf_L}$ vanish. Indeed, if one of them was non-zero, then the HFB Hamiltonian $\opH_{HFB}$ would include terms which would break $U(1)$ gauge invariance, such as $\int dx \, m(x) \, \psi^*(x)\psi^*(x)+h.c.$\,. The quasifree states we consider are thus characterized by their truncated expectation $\gamma_L$, and we will replace the variable $\qf_L$ by $\gamma_L$ in the sequel of this section.

We use the expression of the HFB Hamiltonian \eqref{eq-HHFB-def-1} with $v=g\delta$ (and $\phi=0$, $\sigma=0$), although this expression was derived for more regular interaction potentials~$v$'s:
\eqn 
    \opH_{HFB}(\qf_L) = \int dxdy\,  \psi^*(x)\psi(y) \, (-\Delta +gn)(x;y) \,,
\eeqn
with $n=n(x)=\gamma_L(x;x)$. The translation invariance implies that the kernel $\gamma_L(x;y)$ is a function of $x-y$, that we still denote by $\gamma_L$, and therefore $n=n(x)=\gamma_L(x;x)$ is  independent of $x$. 
%

Applying the fixed point equation \eqref{eqn:fixed-point-Gibbs-quasifree-state} with $\opA=\psi^*(y)\psi(x)$ one can express it equivalently in the variable $\gamma_L$:
\eqn\label{eq-tGammma-fixpt-1}
    \gamma_L = 
    \frac{1}{\exp( \, \beta ( \, -\Delta+g n \1 -\mu\1) \, ) \, - \, \1} \,,
\eeqn 
for $n\in[0,\infty)$. 
The operator $\gamma_L $ is a pseudodifferential operator with symbol 
\begin{align}
    \label{eqn:coef-gamma-L}
     \hat \gamma_L(k)&:=\int_{\Lambda_L}\gamma_L(x)e^{-ix\cdot k}d x = \frac{1}{\exp(  \beta (  k^2+g n -\mu)  )  -  \1}
\end{align}
of~$\gamma_L$. 
Thus
\begin{align}
\label{eq-nL-hatgamma}
    n &=\gamma_L(0)
    = 
    \frac{1}{|\Lambda_L|}\sum_{k\in\Lambda_L^*} \hat\gamma_L(k) \,.
\end{align}
As the Fourier coefficients of $\gamma_L$ depend only of
the number $n$, we 
obtain from \eqref{eq-tGammma-fixpt-1}, \eqref{eqn:coef-gamma-L} and \eqref{eq-nL-hatgamma} a \textit{nonlinear fixed point equation} for $n$:
\eqn\label{eq-nm-fixedpt-1}
    n
    =     
    \frac{1}{|\Lambda_L|}\sum_{k\in\Lambda_L^*} 
     \frac{1}{\exp(  \beta ( k^2+gn -\mu )  )  -  1} \,.
\eeqn
Note that the knowledge of $n$ satisfying \eqref{eq-nm-fixedpt-1}, or of $\gamma_L$ satisfying \eqref{eq-tGammma-fixpt-1}  or of $\qf_L$ satisfying \eqref{eqn:fixed-point-Gibbs-quasifree-state} are equivalent.

From a physical point of view, it is natural to fix the density $n$, which can be tuned in an experiment and to compute $\mu$. So $n$ will be a parameter and we will solve \eqref{eq-nm-fixedpt-1} with the unknown $\mu$.

\begin{lemma}\label{lem:existence-mu-L}
Let $g,\beta,n>0$, and, for $d\geq 3$. Let $n_c$ be the critical density
\eqn
    n_c := \frac{1}{(2\pi)^d}\int_{\mathbb R^d} \frac{dk}{e^{\beta k^2}-1}=\frac{\zeta(\frac{d}{2})\Gamma(\frac{d}{2})}{(2\pi)^d}\beta^{-\frac{d}{2}}\,,
\eeqn
where $\zeta(x)=\sum_{n\geq 1}n^{-x}$ and $\Gamma(x)=\int_0^\infty t^{x-1}e^{-t}dt$.

We define $S_L:(-\infty,gn)\to \mathbb R$ and $S_\infty:(-\infty,gn]\to \mathbb R$ through
\begin{align*}
S_L(\mu)&:=\frac{1}{|\Lambda_L|}\sum_{k\in\Lambda_L^*} 
     \frac{1}{\exp( \beta ( k^2+gn -\mu ) ) -  1} \,,\\
S_\infty(\mu) &:= 
    \frac{1}{(2\pi)^d}\int_{\mathbb R^d} 
     \frac{dk}{\exp(  \beta (  k^2 + gn -\mu )  )  -  1} \,.
\end{align*}
Then:
\begin{itemize}
\item There exists a unique $ \mu_L(n) < gn$ such that \eqref{eq-nm-fixedpt-1} holds, i.e., 
\eqn \label{eqn:mu-L}
n=S_L(\mu_L(n)) \,.
\eeqn
\item If $n < n_c$, there exists a unique $\mu_\infty(n) < gn$ such that
\eqn\label{eqn:mu-infty}
    n    =   S_\infty(\mu_\infty(n))  \,.
\eeqn
We extend the function $\mu_\infty $ to $(0,\infty)$ by setting $\mu_\infty(n)=gn$ for $n\geq n_c$.
\end{itemize}
\end{lemma}
\begin{remark}
The critical density $n_c$ can be explicitly computed.
\end{remark}
\begin{proof} In the discrete case, the existence follows from the intermediate value theorem because the map $S_L$ 
is continuous with limits $0$ at $-\infty$ and  $\infty$ at $gn$. The map $S_L$ is strictly increasing and thus there exists a unique $\mu_L(n)$ such that $n=S_L(\mu_L(n))$.

In the continuous case, we first prove the existence of $\mu_\infty(n)$, for a given $n>0$, the map
$(0,gn]\ni\mu \mapsto S_\infty(\mu)$
is well defined, continuous, $\lim_{\mu\to-\infty}S_\infty(\mu)=0$, $S_\infty(gn)=n_c$, and thus the intermediate value theorem yields the existence of a $\mu_\infty$ satisfying \eqref{eqn:mu-infty}. Since $S_\infty$ is strictly increasing, the uniqueness follows.
\end{proof}

In Theorem~\ref{thm:existence-bose-condensate}, we prove that the thermodynamic limit $\gamma_\infty$ of the self-consistent equation
\eqref{eq-tGammma-fixpt-1} for $\gamma_L$ is well defined and exhibits the so called Bose-Einstein condensation.

\begin{theorem}\label{thm:existence-bose-condensate}
Let $g,\beta,n>0$ and $d\geq 3$. Let $\gamma_L$, $n_c$, $\mu_L$ and $\mu_\infty$ as defined in \eqref{eq-tGammma-fixpt-1} and  Lemmata~\ref{lem:existence-mu-L}. Then
\eqn
\mu_L(n) \xrightarrow[L \to \infty]{} \mu_\infty(n)
 \quad \text{and} \quad
\gamma_L \xrightarrow[L \to \infty]{\cD'} \gamma_\infty     \,,
\eeqn
where  
\eqn
\hat\gamma_\infty(k)= \max\{0,n-n_C\} \, \delta(k) + 
\frac{1}{\exp\big(\beta (k^2+gn-\mu_\infty(n))\big)-1}\,.
\eeqn
\end{theorem}
\begin{remark}
The presence of the $\delta(k)$ term is interpreted as the existence of Bose-Einstein condensation, because there is an accumulation of particles in the zero mode. It occurs when $\beta^{d/2}n\geq C_d$ with $C_d$ a constant depending only on the dimension.
\end{remark}
\begin{proof}[Proof of Theorem~\ref{thm:existence-bose-condensate}]
First we prove the convergence of $\mu_L(n)$ towards a $\mu_\infty(n)$.
We first remark that $\mu_L(n)\geq -C$ for some constant $C>0$ independent of $L$. (Otherwise one could extract a subsequence such that $n=S_{L_j}(\mu_{L_j}(n)) \to 0<n$.) Thus the accumulation points of $\mu_L(n)$ are contained in $[-C,gn]$. 
Let $\mu_{L_j}(n)$ denote an extracted sequence converging to an accumulation point $\mu'$.

In the case $n<n_c$: 
If $\mu'=gn$ then for $j$ large enough $\mu_{L_j}(n)\geq (\mu_\infty(n)+gn)/2$, thus 
\begin{multline*}
n  = S_{L_j}(\mu_{L_j}(n))\geq S_{L_j}\Big(\frac{gn + \mu_\infty(n)}{2}\Big)
     \to S_\infty\Big(\frac{gn + \mu_\infty(n)}{2}\Big)>S_\infty(\mu_\infty(n)) =n
\end{multline*}
and which would lead to a contradiction. Note that it is crucial that  $\mu_\infty(n)<gn$ for $n<n_c$ to get the convergence to the integral $S_\infty\big(\frac{gn + \mu_\infty(n)}{2}\big)$. It thus follows that $\mu'<gn$. Then $S_{L_j}(\mu_{L_j}(n))$ 
converges to $n$, because by definition of $\mu_L(n)$ this sum is equal to $n$, and also to $S_\infty(\mu')$. 
(One has to control the dependency in $\mu_{L_j}(n)$ in the Riemann sums.) Hence $\mu'=\mu_\infty(n)$ and the unique accumulation point is $\mu_\infty(n)$. We thus proved the convergence of $\mu_L(n)$ to $\mu_\infty(n)$.

In the case $n\geq n_c$, we sketch an argument similar to the one above. If an accumulation point $\mu'$ was such that $\mu'< gn$, then the sums $S_{L_j}(\mu_{L_j}(n))$ would converge to integrals with a value strictly smaller than $n_c$ and thus strictly smaller than $n$. This would lead to a contradiction. Thus the only possible accumulation point is $gn$ and $\mu_L(n)\to gn=\mu_\infty(n)$.

We now prove the convergence of $\gamma_L$ towards $\gamma_\infty$. Let $\varphi\in C^\infty_0 (\mathbb R^d)$. For $L$ large enough the support of $\varphi$ is included in $\Lambda_L$, and
\eqn
 \int_{\Lambda_L} \gamma_L \varphi
 = \frac{1}{|\Lambda_L|}\sum_{k\in \Lambda_L^*} \hat \gamma_L(k) \hat \varphi(k) \,.
\eeqn
On the other hand $\langle \gamma_\infty,\varphi\rangle_{\mathcal D'}=\langle \gamma_\infty,\varphi\rangle_{\mathcal S'}=\langle \hat\gamma_\infty,\hat\varphi\rangle_{\mathcal S'}$
(Note that in the normalization we choose, the Fourier coefficients of $\varphi$ on $\Lambda_L$ and the Fourier transform coincide, there is thus no need to specify the hat notation.) The convergence of $\gamma_L$ to $\gamma_\infty$ is thus equivalent to
\eqn
 \frac{1}{|\Lambda_L|}\sum_{k\in \Lambda_L^*} \hat \gamma_L(k) \hat \varphi(k) \to  \max\{0,n-n_c\}\hat\varphi(0) +
\int_{\mathbb R^d}\frac{(2\pi)^{-d}\hat\varphi(k)dk}{e^{\beta (k ^2+gn-\mu_\infty(n))}-1} \,.
\eeqn
for all $\varphi$.

In the case $n<n_c$ the convergence is thus just a convergence of Riemann sums of the integral (with the small additional difficulty that $\mu_{L}(n)$ depends on $L$ in the sum) because there is no singularity in the function $k \mapsto (\exp(\beta(k^2+gn-\mu_\infty(n)))-1)^{-1}$.

In the case $n\geq n_c$: Let $\varepsilon>0$. 
First note that, for any fixed $\eta>0$
\eqn
 \frac{1}{|\Lambda_{L}|}\sum_{\substack{k\in\Lambda_L^*\\ |k|>\eta}} 
     \frac{\hat\varphi(k)}{\exp(  \beta (  k^2+gn -\mu_L(n) )  )  -  1} 
     \to
\int_{|k|> \eta}\frac{(2\pi)^{-d}\hat\varphi(k)dk}{e^{\beta k ^2}-1} \,, \nonumber
\eeqn
as $L\to\infty$. We choose $\eta>0$ small enough so that
\eqn
 |k|\leq \eta \, \Rightarrow \,|\hat\varphi(k)-\hat\varphi(0)|
 \leq \frac{\varepsilon}{4n}
\qquad
\text{and}
\qquad 
\int_{|k|\leq \eta}\frac{(2\pi)^{-d}\hat\varphi(0)dk}{e^{\beta k ^2}-1} \leq \frac{\varepsilon}{4} \,.\nonumber
\eeqn
The first condition on $\eta$ yields
\begin{align*}
\Bigg|\frac{1}{|\Lambda_{L}|}\sum_{\substack{k\in\Lambda_L^*\\ |k|\leq \eta}} 
     \frac{\hat\varphi(k)-\hat\varphi(0)}{\exp(  \beta (  k^2+gn -\mu_L(n) )  )  -  1}\Bigg|
\leq\frac{\varepsilon}{4} \,,
\end{align*}
then, the second condition on~$\eta$ implies
\begin{align*}
\limsup_{L\to\infty}\Bigg|\frac{1}{|\Lambda_{L}|}\sum_{\substack{k\in\Lambda_L^*\\ |k|\leq \eta}} 
     \frac{\hat\varphi(0)}{\exp(  \beta (  k^2+gn -\mu_L(n) )  )  -  1} - (n-n_c)\hat\varphi(0)\Bigg|\leq \frac{\varepsilon}{4} \,.
\end{align*}
Hence
\begin{multline*}
\limsup_{L\to\infty}\Bigg|\frac{1}{|\Lambda_{L}|}\sum_{\substack{k\in\Lambda_L^*}} 
     \frac{\hat\varphi(k)}{\exp(  \beta (  k^2+gn -\mu_L(n) )  )  -  1} \\ 
     - (n-n_c)\hat\varphi(k) - 
\int_{\mathbb R^d}\frac{(2\pi)^{-d}\hat\varphi(k)dk}{e^{\beta k ^2}-1}  \nonumber\Bigg|\leq \varepsilon \,,
\end{multline*}
and as this holds for any $\varepsilon>0$, we get the result.
\end{proof}

\appendix


\section{Self-adjointness of the Hamiltonian $\opH$}
\label{sec:Self-adjointness_H}

We use that the pair potential $v$ is infinitesimally $\Delta-$bounded,
i.e. for any $\varepsilon \in (0,1]$,
\eqn
v 
\leq -\varepsilon\Delta + C\varepsilon^{-1}\,
\eeqn
(we write $C$ for constants which depend on $v$, $d$ and change along the estimates) 
to obtain after taking $\varepsilon = 1/(3(n-1))$,   
\eqn
v(x-y)\leq \frac{1}{6(n-1)} (-\Delta_x-\Delta_y) + C (n-1) \,.
\eeqn 
Then, summing the $n(n-1)/2$ terms of this form on each $n$-particles subspace of the Fock space, we obtain that
\eqn \label{eqn:estimate-dGamma2-v2}
\mathbb V := \frac{1}{2}\int dx dy \; v(x-y)\psi^*(x) \psi^*(y)
    \psi(x) \psi(y) \leq  
\frac23\mathbb T+C\opN^{3} \,.
\eeqn
\DETAILS{Hence, from the definition of $\opH$, \eqref{eqn:estimate-dGamma-V} and \eqref{eqn:estimate-dGamma2-v} we get
\begin{align}
\mathbb T \leq 3\opH + C \opN^{2} \,.
\end{align}

(revised) 
The same arguments as those used to prove~\eqref{eqn:estimate-dGamma2-v} allow to deduce
\eqn
\mathbb V \leq C(\mathbb T + \mathbb N^3)
\eeqn}
for some $C >0$, with  $\mathbb T$  defined in~\eqref{eqn:def-mathbb-T}. 
 \DETAILS{from 
\eqn
|v| 
\leq -\varepsilon\Delta + C\varepsilon^{-1}\,.
\eeqn
for some $C >0$.} 
One can then use the KLMN theorem and the Nelson theorem (see \cite{ReedSimonII1980,Nelson1972}) to prove the self-adjointness of $\opH$. (Details can be adapted from, e.g., \cite[Section 3]{AmmariBreteaux2012}.)

\section{Definition of quasifree states} 
\label{app-quasifreedef}

For brevity, we write $\psi_j^\sharp :=
\psi^\sharp(x_j)$.
We recall that the truncated
expectations are defined via
\begin{equation}
\label{eqn:quasi-wick-2} 
    \omega(\psi_1^\sharp\cdots
    \psi_n^\sharp) = \sum_{P_n} \prod_{J \in P_n}
   \om^T (\prod_{j \in J}\psi_{j}^\sharp),
\end{equation}
where $P_n$ are partitions of the ordered set $\{1, ..., n\}$ into
ordered subsets. 

We have $ \mu(\psi) =
\omega(\psi)$ and
\begin{equation}\label{2term-wick-2}
    \om^T(\psi_1^\sharp\psi_2^\sharp) =
    \omega(\psi_1^\sharp \psi_2^\sharp)-
    \omega(\psi_1^\sharp)\omega(\psi_2^\sharp).
\end{equation} 
For   quasifree states, the correlation functions
$\omega(\psi_1^\sharp\cdots \psi_n^\sharp)$, with $n
>2$ can be expressed through $\omega(\psi^\sharp(x))$ and
$\omega(\psi^\sharp(x)\psi^\sharp(y))$ according to
the Wick formula. For example, 
\begin{equation}\label{3term-wick} \omega(\psi_1^\sharp \psi_2^\sharp
\psi_3^\sharp) =
\omega(\psi_1^\sharp)\omega(\psi_2^\sharp
\psi_3^\sharp) +
\omega(\psi_2^\sharp)\omega(\psi_1^* \psi_3^\sharp)
+ \omega(\psi_3^\sharp)\omega(\psi_1^\sharp
\psi_2^\sharp)-2\prod_{i=1}^{3}\omega(\psi_i^\sharp)
\end{equation}
and
\begin{equation}\label{4term-wick}
\omega(\psi_1^\sharp \psi_2^\sharp
\psi_3^\sharp \psi_4^\sharp) =
\omega(\psi_1^\sharp
\psi_2^\sharp)\omega(\psi_3^\sharp
\psi_4^\sharp) +
\omega(\psi_1^\sharp\psi_3^\sharp)\omega(\psi_2^\sharp
\psi_4^\sharp) +
\omega(\psi_1^\sharp\psi_4^\sharp)\omega(\psi_2^\sharp
\psi_4^\sharp)
-2\prod_{i=1}^{4}\omega(\psi_i^\sharp)
\end{equation}
(remember that $\psi$'s stand on the right of
$\psi^*$'s.) Note that
$$\omega(\psi^*(x)) = \overline{\omega(\psi(x))},\
\omega(\psi_1^*\psi_2^*) =
\overline{\omega(\psi_2\psi_1)}$$ and
$$\omega(\psi_1 \psi_2^*) =
\omega(\psi_2^*\psi_1)+\delta(x-y).$$ Thus a
quasifree state $\omega$ is completely determined by the functions
$\omega( \psi(x))$, $\mu( \psi^*(x)
\psi(y))$ and $\mu( \psi(x)  \psi(y))$.

\begin{remark} It is instructive to rewrite correlation functions
for a quasifree state $\omega$ in terms of the fluctuation fields
$\chi(x)$ which are defined as follows
\begin{equation} \label{fieldDecomp}
 \psi = \phi + \chi,\ \mbox{where}\
\phi(x)=\omega(\psi(x)),
\end{equation}
the average field. 
Then $\omega$ is  a quasifree state iff $\omega(\chi_1^\sharp\cdots
\chi_{2n-1}^\sharp)=0$ and $$\omega(\chi_1^\sharp\cdots
\chi_{2n}^\sharp)=\sum_{\pi \in S_n} \prod_{i=1}^{2n-1}
\omega(\chi_{\pi(i)}^\sharp \chi_{\pi(i+1)}^\sharp),$$ where the sum
is taken over all the permutations $\pi$ of the set of indices $\{1,
...,2n\}$ satisfying $\pi(1) < ... < \pi(2n)$.
\end{remark}

\section{Derivation of the bosonic HFB equations}
\label{sec-HFB-deriv-1}
 
In this section, we prove Theorem \ref{prop-HFB-vfull-1}.  The derivations below are done in a somewhat informal way commonly used in dealing with operators on Fock spaces (see e.g. \cite{Berezin1966, BratteliRobinson-II-1996, GustafsonSigal2011}). For instance, the commutator  
$[A, H]$, for $A= \psi(x)$ and $A= \psi(x) \psi(y)$, contains the terms  $\Delta_x \psi(x)$ and $\psi(x) \Delta_y \psi(y)$. 
 The formal computation gives $\qf(\Delta_x \psi(x))=\Delta_x \qf(\psi(x))$ and $\qf(\psi(x) \Delta_y \psi(y))=\Delta_y \qf(\psi(x) \psi(y))$, which are well-defined by our assumptions and are equal to $\Delta_x \phi(x)$ and $\Delta_y \s (x, y)$, respectively.

 To do this more carefully, one uses, instead of operator functions $\psi^\#(x)$, the operator functionals $\psi^\#(f)$, for some nice $f$.  E.g., instead  $[\psi(x), H]$, we consider the commutator  $[\psi(f), H]$, for any nice $f$, and concentrate on the term $\psi(\Delta f)$ it contains. Clearly, $\qf$ is well defined on $\psi(\Delta f)$ and can be written as $\qf(\psi(\Delta f))=\int \overline{\Delta f} (x)\qf(\psi(x))=\int \Delta \overline{ f} (x)\phi(x)=\int \overline{ f} (x) \Delta\phi(x)$. Thus we obtain the same result as above but in a weak form.

\begin{proof}[Proof of Theorem \ref{prop-HFB-vfull-1}]
We first observe that the three following condition are equivalent: 
\begin{enumerate}
\item A quasifree state $\om_{t}^{q}$ satisfies
\begin{align}
i\partial_{t}\om_{t}^{q}\big(\opA\big) & =\omega_{t}^{q}\big([\opA,\opH]\big)\,,
\end{align}
for any operator $\opA$ of order $\le 2$ in the fields.
\item A quasifree state $\om_{t}^{q}$  satisfies
\begin{align}
    i\partial_{t}\om_{t}^{q}\big(\psi(x)\big) & =
    \omega_{t}^{q}\big([\psi(x),\opH]\big)\,,\\
    i\partial_{t}\om_{t}^{q}\big(\psi^{*}(y)\psi(x)\big) & =
    \omega_{t}^{q}\big([\psi^{*}(y)\psi(x),\opH]\big)\,,\\
    i\partial_{t}\om_{t}^{q}\big(\psi(x)\psi(y)\big) & =
    \omega_{t}^{q}\big([\psi(x)\psi(y),\opH]\big)\,.
\end{align}
\item A quasifree state $\om_{t}^{q}$ with truncated expectations $\phi_{t}$, $\gamma_{t}$
and $\sigma_{t}$ satisfies 
\begin{align}
\label{eq1} i\partial_{t}\phi_{t}(x) & =\om_{t}^{q}([\psi(x),\opH]) \ ,\\
\label{eq2} i\partial_{t}\gamma_{t}(x;y) & =\om_{t}^{q}([\psi^{*}(y)\psi(x),\opH])-i\partial_{t}\big(\phi_{t}(x)\overline{\phi(y)}\big) \ ,\\
\label{eq3} i\partial_{t}\sigma_{t}(x,y) & =\om_{t}^{q}([\psi(x)\psi(y),\opH])-i\partial_{t}\big(\phi_{t}(x)\phi(y)\big) \ .
\end{align}
\end{enumerate}
We now suppose $\om_{t}^{q}$ satisfies \eqref{eq1} - \eqref{eq3}. 
Using the definition of the Hamiltonian, we obtain 
\begin{align}
i\partial_{t}\phi_{t}(x) & =\om_{t}^{q}\Big(\big[\psi(x),\int\psi^{*}(y)h(y;y')\psi(y')\, dydy'\big] \nonumber \\
 & \quad+\frac{1}{2}\big[\psi(x),\int v(y-y')\psi^{*}(y)\psi^{*}(y')\psi(y')\psi(y)\, dydy'\big]\Big)\\
 & =\om_{t}^{q}\Big(\int h(x;y')\psi(y')\, dy'+\int v(x-y)\psi^{*}(y)\psi(y)\psi(x)\, dy\Big)\,,
\end{align}
where we used the CCR \eqref{commrel} to get
\begin{multline}
\big[\psi(x),\psi^{*}(y)\psi^{*}(y')\psi(y')\psi(y)\big]\\
=\delta(x-y)\psi^{*}(y')\psi(y')\psi(y)+\delta(x-y')\psi^{*}(y)\psi(y)\psi(y')\,.\label{eq:commutator_psi,psi-st_psi-st_psi_psi}
\end{multline}
As $\om_{t}^{q}$ is a quasifree state (see Appendix~\ref{app-quasifreedef})
\begin{multline}
\om_{t}^{q}\big(\psi^{*}(y)\psi(y)\psi(x)\big)\\
=|\phi_{t}(y)|^{2}\phi_{t}(x)+\sigma(y;x)\bar{\phi}_{t}(y)+\phi_{t}(x)\gamma(y;y)+\phi_{t}(y)\gamma(x;y)\ .
\end{multline}
We thus deduce that
\begin{align*}
i\partial_{t}\phi_{t}(x) & =\int h(x;y')\phi_{t}(y')\, dy'\\
& \quad +\int v(y-x)\phi_{t}(x)\gamma_{t}(y;y)\, dy+\int v(y-x)\phi_{t}(y)\gamma_{t}(x;y)\, dy\\
 & \quad+\int v(x-y)\sigma_{t}(y,x)\bar{\phi}_{t}(y)\, dy+\int v(y-x)\phi_{t}(y)\phi_{t}(x)\bar{\phi}_{t}(y)\, dy\\
 & =\big((h+b [\gamma_{t}])\phi_{t}\big)(x)+ k (\s^{\phi_t}_t)\bar{\phi}_{t}(x)
\end{align*}
which is the dynamical equation \eqref{eq:BHF-phi} for $\phi_{t}$.

For $\gamma_{t}$ and $\sigma_{t}$, instead of $\om_{t}^{q}$ we use 
\eqn\label{omC}
    \om_{C,t}^{q}(\mathbb{A}) :=\om_{t}^{q}(W_{\phi_{t}}\opA W_{\phi_{t}}^{*}) \,,
\eeqn
where, recall, $W_{\phi}=\exp\big(\psi^{*}(\phi)-\psi(\phi)\big)$, the Weyl operators, which satisfy 
\begin{align}
    W_{\phi}^*\psi(x)W_{\phi}=\psi(x)+\phi(x) \, .
\end{align}
Note that the state $\om_{C,t}^{q}$ is quasifree because $\om_{t}^{q}$
is quasifree. By construction $\om_{C,t}^{q}(\psi(x))=0$ and thus using \eqref{eqn:quasi-wick} and the quasifreeness of $\om_{C,t}^{q}$ one sees that 
$\om_{C,t}^{q}$ vanishes on monomials of odd order in the fields. 
This provides substantial simplifications in the computations
below.

In particular the equations of the dynamics for $\gamma_t$ and $\sigma_t$ can be rewritten
\begin{align}
i\partial_{t}\gamma_{t}(x;y) & =\om_{C,t}^{q}([\psi^{*}(y)\psi(x),W_{\phi_{t}}^{*}\opH W_{\phi_{t}}]) \ , \label{eq:deriv_gamma_omega_centered}\\
i\partial_{t}\sigma_{t}(x_1,y) & =\om_{C,t}^{q}([\psi(x_1)\psi(y),W_{\phi_{t}}^{*}\opH W_{\phi_{t}}]) \ .\label{eq:deriv_sigma_omega_centered}
\end{align}
We compute $W_{\phi_{t}}^{*}\opH W_{\phi_{t}}$ modulo terms of odd degree and of degree 0 in the creation and annihilation operators: 
\begin{align}
W_{\phi_{t}}^{*}\opH W_{\phi_{t}} & \equiv\int\psi^{*}(z)\big(h+b_v[|\phi\rangle\langle\phi|]\big)(z;z')\psi(z')\, dzdz'\nonumber\\
 & \quad+\frac{1}{2}\int v(z-z')\phi_{t}(z)\phi_{t}(z')\psi^{*}(z)\psi^{*}(z')\, dzdz'+adj.\nonumber\\
 & \quad+\frac{1}{2}\int v(z-z')\psi^{*}(z)\psi^{*}(z')\psi(z')\psi(z)\, dzdz' \ .
\end{align}
Because $\om_{C,t}^{q}$ vanishes on monomials of odd order in the fields and using the commutator, the knowledge of $W_{\phi_{t}}^{*}\opH W_{\phi_{t}}$ modulo terms of odd degree and of degree 0 in the creation and annihilation operators is sufficient to compute the time derivative \eqref{eq:deriv_gamma_omega_centered} of $\gamma_t$. Thus using the CCR we get
\begin{align}
i\partial_{t}\gamma_{t}(x;y) & =\int\om_{C,t}^{q}\Big(\big(h+b_v[|\phi_{t}\rangle\langle\phi_{t}|]\big)(x;z)\psi^{*}(y)\psi(z) \nonumber \\
 & \quad-\big(h+B_v[|\phi_{t}\rangle\langle\phi_{t}|]\big)(z;y)\psi^{*}(z)\psi(x) \nonumber \\
 & \quad+v(z-x)\phi_{t}(z)\phi_{t}(x)\psi^{*}(y)\psi^{*}(z)-v(z-y)\overline{\phi_{t}(z)\phi_{t}(y)}\psi(z)\psi(x) \nonumber \\
 & \quad+v(z-x)\psi^{*}(y)\psi^{*}(z)\psi(x)\psi(z)-v(z-y)\psi^{*}(z)\psi^{*}(y)\psi(z)\psi(x)\Big)\, dz\,.
\end{align}
From the quasifreeness of $\om_{C,t}^{q}$ follows
\begin{align}
i\partial_{t}\gamma_{t}(x;y) & =\big[h+b_v[|\phi_{t}\rangle\langle\phi_{t}|+\gamma_{t}],\gamma_{t}\big](x;y) \nonumber \\
 & \quad+\int\big(v(z-x)\phi_{t}(z)\phi_{t}(x)\overline{\sigma_{t}(y,z)}-v(z-y)\overline{\phi_{t}(z)\phi_{t}(y)}\sigma_{t}(z,x) \nonumber \\
 & \quad+v(z-x)\sigma_{t}(x,z)\overline{\sigma_{t}(y,z)}-v(z-y)\sigma_{t}(x,z)\overline{\sigma_{t}(y,z)}\big)\, dz\,.
\end{align}
which is the dynamical equation  \eqref{eq:BHF-gamma} for $\gamma_{t}$.

Using the same arguments as for $\gamma_{t}$, we get
\begin{align}
i\partial_{t} & \sigma_{t}(x;y)  \nonumber \\
& =\om_{C,t}^{q}\Big(v(x-y)\phi_{t}(x)\phi_{t}(y)+v(x-y)\psi(x)\psi(y) \nonumber \\
 & \quad+\int\big((h+b_v[|\phi_{t}\rangle\langle\phi_{t}|])(x;z)\psi(y) +(h+b_v[|\phi_{t}\rangle\langle\phi_{t}|])(y;z)\psi(x)\big)\psi(z) \nonumber \\
 & \quad+v(x-z)\psi^{*}(z)\psi(y)\phi_{t}(x)\phi_{t}(z)+v(y-z)\psi^{*}(z)\psi(x)\phi_{t}(y)\phi_{t}(z) \nonumber \\
 & \quad+v(x-z)\psi^{*}(z)\psi(y)\psi(x)\psi(z)+v(y-z)\psi^{*}(z)\psi(x)\psi(y)\psi(z)\big)\, dz\Big)\,.
\end{align}
From the quasifreeness of $\om_{C,t}^{q}$ follows
\begin{align}
i\partial_{t} & \sigma_{t}(x;y)  \nonumber \\
 & =v(x-y)\phi_{t}(x)\phi_{t}(y)+v(x-y)\sigma_{t}(x,y) \nonumber \\
 & \quad+\int\Big(\big(h+b_v[|\phi_{t}\rangle\langle\phi_{t}|]\big)(x;z)\sigma_{t}(y,z)+\big(h+b_v[|\phi_{t}\rangle\langle\phi_{t}|]\big)(y;z)\sigma(x,z) \nonumber \\
 & \quad+v(x-z)\gamma_{t}(y;z)\phi_{t}(x)\phi_{t}(z)+v(y-z)\gamma_{t}(x;z)\phi_{t}(y)\phi_{t}(z) \nonumber \\
 & \quad+v(x-z)\big(\gamma_{t}(x;z)\sigma_{t}(z,y)+\gamma_{t}(y;z)\sigma_{t}(z,x)+\gamma_{t}(z;z)\sigma_{t}(x,y)\big) \nonumber \\
 & \quad+v(y-z)\big(\gamma_{t}(x;z)\sigma_{t}(z,y)+\gamma_{t}(y;z)\sigma_{t}(z,x)+\gamma_{t}(z;z)\sigma_{t}(x,y)\big)\Big)\, dz\,,
\end{align}
which is the dynamical equation  \eqref{eq:BHF-sigma} for $\sigma_{t}$.
\end{proof}

\section{Equivalence of the HBF equations with \\
the evolution generated by $\HHFB(\qf_t)$}
\label{sec-HFB-equiv-1}

In this section, we prove Theorem \ref{thm-Hquadr-dyn-1}. 

Let a quasifree state $\om_{t}^{q}$ satisfy \eqref{eq-omt-commut-3} and let $\phi_{t}, \gamma_{t}$ and $\sigma_{t}$ denote its truncated expectations. Below,  we use the  abbreviations $h (t)\equiv h (\gamma_t^{\phi_{t}})$ and $k (t)\equiv  k (\sigma^{\phi_{t}}_t),$ where, recall, $\gamma^{\phi}:=\gamma+|\phi\rangle\langle\phi|$ and $\sigma^{\phi}:=\sigma +|\phi\rangle\langle\phi|$, and 
$ h (\gamma)$ and $  k (\sigma)$ are defined in \eqref{h} and \eqref{k}.
To find the equation for $\phi_{t}$, we compute \begin{align*}
i\partial_{t} \phi_{t}(x) & = {\om}_{t}^{q}\big([\psi(x), \Hhfb]\big)\\
 & =\tilde{\om}_{t}^{q}\Big(\int h (t)(x;z)\psi(z)dz-b[|\phi_{t}\rangle\langle\phi_{t}|]\phi_{t}(x) +\int\psi^{*}(z)k (t)(x,z)\, dz\Big)\\
 & =h (t) {\phi}_{t}(x)-b [|\phi_{t}\rangle\langle\phi_{t}|]\phi_{t}(x)+k (t)\overline{{\phi}_{t}}(x)\,.
\end{align*}
 Hence 
$\phi_{t}$ satisfies \eqref{eq:BHF-phi}.

For ${\gamma}_{t}$ and ${\sigma}_{t}$ we remark that,
modulo terms of order one and constants $W_{\phi_{t}}^{*} \Hhfb W_{\phi_{t}}$
and $\Hhfb$ coincide, hence
\begin{align}
W_{\phi_{t}}^{*}  \Hhfb W_{\phi_{t}} \equiv & \int h (t)(z;z')\psi^{*}(z)\psi(z')\, dzdz' \nonumber \\
   &+\frac{1}{2}\int\psi^{*}(z_{1})\psi^{*}(z_{2})k (t)(z_1,z_2)\, dz_{1}dz_{2}+adj.\,.
\end{align}
Recall the definition \eqref{omC} of $\om_{C,t}^{q}(\mathbb{A})$. As in the proof of Theorem \ref{prop-HFB-vfull-1} the terms coming
from the derivative of $W_{\phi_{t}}$ simplify:
\[
i\partial_{t} {\gamma}_{t}(x;y)= {\om}_{C,t}^{q}\big([\psi^{*}(y)\psi(x),W_{\phi_{t}}^{*}\Hhfb W_{\phi_{t}}]\big)\,.
\]
It is sufficient to consider $W_{\phi_{t}}^{*}\tilde{\opH}(\phi_{t},\gamma_{t},\sigma_{t})W_{\phi_{t}}$ modulo monomials of odd order in the fields:
\begin{align*}
i\partial_{t}\gamma_{t}(x;y) & = {\om}_{C,t}^{q}\Big(\int h (t)(x;z)\psi^{*}(x)\psi(z)dz-\int h (t)(z;y)\psi^{*}(z)\psi(y)\, dz\\
 & \quad+\int\psi^{*}(z)\psi^{*}(y)k (t)(z,x)\, dz -\int\overline{k (t)(z,y)}\psi(z)\psi(x)\, dz\Big)\\
 & =\int h (t)(x;z)\gamma_{t}(z;x)dz-\int {\gamma}_{t}(x;z)h_{v}(t)(z;y)\, dz\\
 & \quad+\int\overline{\sigma_{t}(y,z)}k (t)(z,x)\, dz -\int\overline{k (t)(z,y)}\sigma_{t}(x,z)\, dz\Big)\,.
\end{align*}
Similarily
\begin{align}
i\partial_{t}\sigma_{t}(x;y)= {\om}_{C,t}^{q}\big([\psi(x)\psi(y),W_{\phi_{t}}^{*}\Hhfb W_{\phi_{t}}]\big)
\end{align}
and
\begin{multline}
i\partial_{t}\gamma_{t}(x;y)  =\int h (t)(x;z)\sigma_{t}(x,z)dz+\int h (t)(y;z)\sigma_{t}(y,z)\, dz  \\
  \quad+\int\gamma_{t}(y,z)k (t)(z,x)\, dz +\int\gamma_{t}(x,z)k (t)(z,y)\, dz +k (t)(x,y)
\end{multline}
Thus 
${\gamma}_{t}$ and ${\sigma}_{t}$ satisfy \eqref{eq:BHF-gamma} and \eqref{eq:BHF-sigma}. 

We have shown that, if  a quasifree state $\om_{t}^{q}$ satisfies \eqref{eq-omt-commut-3}, then its truncated expectations, $\phi_{t}, \gamma_{t}$ and $\sigma_{t}$, satisfy  \eqref{eq:BHF-phi}, \eqref{eq:BHF-gamma} and \eqref{eq:BHF-sigma}. 
Proceeding in the opposite direction, one shows that, if  truncated expectations, $\phi_{t}, \gamma_{t}$ and $\sigma_{t}$, satisfy  \eqref{eq:BHF-phi}, \eqref{eq:BHF-gamma} and \eqref{eq:BHF-sigma}, then the corresponding quasifree state $\om_{t}^{q}$ satisfies \eqref{eq-omt-commut-3}.
\qquad \qquad \qquad $\Box$



\section{Operators $b$ and $k$} \label{sec:ops-b-k}

Recall that $W^{p, r}(\R^d)$ denotes the standard Sobolev space over $\R^d$. 

 \begin{lemma} \label{prop:continuity_of_B_and_K}
Assume that $v \in W^{p, 1}$ with $p>d$. 
Then, the operators $b$ and $k$ defined in \eqref{h} and \eqref{k}
 possess the following properties:
\begin{enumerate}
\item \label{enu:continuity-B-H1L2-LH1}$b$ is 
continuous from $\cHLd$ to $\cB(\Hd)\simeq M\cB M^{-1}$.
\item \label{enu:continuity-K-H1L2-L2H1} $k$ is 
continuous from $\Hsdd$ to $M^{-1}\cL^2$. 
\end{enumerate}
\end{lemma}

\begin{proof} 
For the detailed proof of statement \eqref{enu:continuity-B-H1L2-LH1},
we refer to~\cite{BoveDaPratoFano1976}. 
For the reader's convenience, we recall here the main arguments.
We first consider the direct term, i.e., the first term in the definition of~$b$.  
It is sufficient to prove that $v*n$ (with functions $n(x)=\gamma(x;x)$) and $\nabla v*n$ uniformly bounded by $\|\gamma\|_{\cHLd}$. As those two bounds are very similar, we focus on the more difficult one, $\nabla v*n$.

Denote by $\tilde\g$  the (generalized) integral kernel of an operator $\g$. 
 Since $v \in W^{p, 1}(\R^d)$ with $p>d$, the function $v$ is bounded. Since $\nabla_x \int_{\R^d} v(x-y) \, \gamma(y;y)dy=\int_{\R^d} v(x-y) \, \nabla_y \gamma(y;y)dy$, we have
\begin{align}
\big\|\nabla_x \int_{\R^d} v(x-y) \, \gamma(y;y)dy\big\|_\infty 
&\leq \big\|v\big\|_\infty \int_{\R^d} |\nabla_y \gamma(y;y)| dy\end{align}
Furthermore, $\int_{\R^d} |\nabla_y \gamma(y;y)| dy\leq  \|\gamma\|_{\cHLd}$, which can proved by using the decomposition $\gamma=\sum_{j=1}^\infty \lambda_j |\varphi_j\rangle\langle\varphi_j|$ with $\lambda_j\geq 0$ of $\gamma$, combined with the Cauchy-Schwarz inequality: 
\begin{align}
\int_{\R^d} |\nabla_y \gamma(y;y)| dy 
&\leq \sum_{j=1}^\infty \lambda_j  \int_{\R^d} |\varphi_j(y)\nabla \varphi_j(y)| dy\\
&\leq \sum_{j=1}^\infty \lambda_j   \| \varphi_j\|_{L^2}  \|\nabla \varphi_j\|_{L^2} \\
&\leq  \sum_{j=1}^\infty \lambda_j  \|M \varphi_j\|^2_{L^2}\leq  \|\gamma\|_{\cHLd}
\end{align}
The last two estimates imply the desired result, $\|\nabla v*n \|_\infty\leq  \|\gamma\|_{\cHLd}$. 
The estimates for the exchange term (the second term in the definition of $B$) are similar.

Point (\ref{enu:continuity-K-H1L2-L2H1}) is equivalent to the estimate
 \begin{align}\label{Mk-est}
\|Mk[\sigma]\|_{\cL^{2}} &\ls \|\s\|_{\cH^{1}_\s},\end{align}  
which we now prove.

Denote   by $\tilde\s$ 
 the (generalized) integral kernel of an operator $\s$. 
Clearly, $\|\sigma\|_{\mathcal{H}^j_\s}\simeq \|\tilde\s\|_{H^{1}}$. 
Denote by $a(x,y)=v(x-y) \tilde\sigma(x, y)$, 
 the integral kernel of $k[\sigma]$. We have that
\begin{align}\label{Mk-est2}
\|Mk\|_{\cL^{2}}^{2} & =\int\int|M_{x}a(x,y)|^2dxdy \le\|a\|_{H^{1}}^{2}.\end{align}
Since $a(x,y)=v(x-y) \tilde\s(x,y)$  and 
\[
\|a \|_{H^{1}}\le \|a \|_{L^{2}}+ \|\p_x a \|_{L^{2}}+ \|\p_y a \|_{L^{2}}\,,
\] 
we use the Leibniz rule,  
$\p_x a(x,y)=(\p v(x-y)) \tilde\s(x,y) + v(x-y) \p_x \tilde\s(x,y)$, 
to find that
\begin{align}
\label{a-estim}
\|a \|_{H^{1}}\le & \big( \|v \|_{L^{\infty}} 
+ \|\p_x v M_x^{-1}\|_{\cB(L^2)} + \|\p_y v M_y^{-1}\|_{\cB(L^2)} \big) 
\, \|\tilde\s \|_{H^{1}} \, ,
\end{align}
where $L^2 := L^2(\R^d_x\times \R^d_y)$. 
The Schwartz and Sobolev inequalities imply that $$\|\p_x v f\|_{L^{2}}\le  \|\p_x v \|_{L^{p}}\|  f\|_{L^{s}}\ls \|v \|_{W^{p, 1}}\| M f\|_{L^{2}},$$ for arbitrary $s$ and $p$ satisfying $\frac1p+\frac1s = \frac12$ and $p>d$. Thus $$\|\p_x v M_x^{-1}\|\ls \|v \|_{W^{p, 1}}\,,$$ and, similarly, $\|\p_y v M_y^{-1}\|\ls  \|v \|_{W^{p, 1}}$. It follows that
\begin{align}\label{a-est'}
\|a \|_{H^{1}}\ls \|v \|_{W^{p, 1}}\|\tilde\s\|_{H^{1}}.\end{align}
This,
 together with \eqref{Mk-est2} and $\|\tilde\s\|_{H^{1}}\simeq \|\sigma\|_{\mathcal{H}^1_\s}$, yields \eqref{Mk-est}. 
\end{proof}

\bibliographystyle{plain}

\end{document}